\documentclass[twocolumn, aps, pra, showpacs, superscriptaddress, floatfix, nofootinbib, 10pt]{revtex4-2}

\usepackage{amsmath,bbm,amsthm}
\usepackage{amssymb}
\usepackage{graphicx}
\usepackage{color}
\usepackage[colorinlistoftodos]{todonotes}
\usepackage[colorlinks=true, allcolors=blue]{hyperref}
\usepackage{dsfont}
\usepackage{float}
\usepackage{cancel}
\usepackage[autostyle]{csquotes}
\usepackage{enumitem}   
\usepackage{mathrsfs}
\usepackage[T1]{fontenc}
\usepackage{soul}
\usepackage{thm-restate}

\newtheorem{remark}{Remark}
\newtheorem{theorem}{Theorem}
\newtheorem{definition}{Definition}
\newtheorem{lemma}{Lemma}

\newtheorem{corollary}{Corollary}
\newtheorem{proposition}{Proposition}

\newtheorem{claim}{Claim}

\newtheorem*{theorem*}{Theorem}

\newcommand{\cB}{\mathcal{B}}
\newcommand{\cH}{\mathcal{H}}
\newcommand{\cK}{\mathcal{K}}
\newcommand{\cN}{\mathcal{N}}
\newcommand{\cR}{\mathcal{R}}
\newcommand{\cS}{\mathcal{S}}

\newcommand{\E}{\mathbb{E}}

\newcommand{\beq}{\begin{equation}}
\newcommand{\eeq}{\end{equation}}
\newcommand{\beqa}{\begin{eqnarray}}
\newcommand{\eeqa}{\end{eqnarray}}

\newcommand{\ket}[1]{\ensuremath{\left|#1\right\rangle}}

\newcommand{\ketbra}[2]{\mbox{$ |#1 \rangle \langle #2 | $}}

\newcommand{\token}{{\mathrm{token}}}
\newcommand{\match}{{\mathrm{color}}}

\def\opone{\leavevmode\hbox{\small1\normalsize\kern-.33em1}}

\newcommand{\Tr}[1]{\mathrm{Tr}\left(#1\right)}

\renewcommand{\today}{\number\day\space\ifcase\month\or
   January\or February\or March\or April\or May\or June\or
   July\or August\or September\or October\or November\or December\fi
   \space\number\year}

\begin{document}

\title{Network Nonlocality via Rigidity of Token-Counting and Color-Matching}
\date{\today}
\author{Marc-Olivier Renou}
\affiliation{ICFO-Institut de Ciencies Fotoniques, 08860 Castelldefels (Barcelona), Spain}
\author{Salman Beigi}
\affiliation{School of Mathematics, Institute for Research in Fundamental Sciences (IPM), P.O.~Box 19395-5746, Tehran, Iran}

\begin{abstract}
\emph{Network Nonlocality} is the study of the \emph{Network Nonlocal} correlations created by several independent entangled states shared in a network. In this paper, we provide the first two generic strategies to produce nonlocal correlations in large classes of networks without input. In the first one, called \emph{Token-Counting} (TC), each source distributes a fixed number of tokens and each party counts the number of received tokens. In the second one, called \emph{Color-Matching} (CM), each source takes a color and a party checks if the color of neighboring sources match. Using graph theoretic tools and Finner's inequality, we show that TC and CM distributions are \emph{rigid} in wide classes of networks, meaning that there is essentially a unique classical strategy to simulate such correlations. Using this rigidity property, we show that certain quantum TC and CM strategies  produce correlations that cannot be produced classicality. This leads us to several examples of Network Nonlocality without input. These examples involve creation of \emph{coherence throughout the whole network}, which we claim to be a fingerprint of genuine forms of Network Nonlocality. 
This work extends a more compact parallel work~\cite{PRL} on the same subject and provides all the required technical proofs. 
\end{abstract}

\maketitle

\section{Introduction}

Bell's theorem shows that local measurements of an entangled state creates Nonlocal quantum correlations~\cite{Bell1964}, a fundamental signature of quantum physics. 
Bell inequalities provide standard methods to characterize the nonlocal correlations of a single quantum state~\cite{BrunnerRMP}, and constitute a powerful toolbox for a large spectrum of applications~\cite{Mayer1998QuantCryptoUnsecAppar,Supic2020SelfTestReview,Kaniewski2016SelfTest,Pironio2010DIQRNG, Herrero2019,Bennett1984BB84, Acin2007DIQKD_CollAttacks}.
Nowadays experiments, such as the first loophole-free violation of a Bell inequality~\cite{HensenLoopholeFree2015}, often involve networks of several sources used to simulate standard single-source Bell protocols. 
Similar procedures are envisioned for long range QKD protocols using multimode quantum repeaters~\cite{Sangouard2011QuantRepeaters}. 

It was recently understood that networks can lead to new forms of nonlocality called \emph{Network Nonlocality}, which is imposed by the network topology.
This remarkable point of view can be exploited for new applications such as certification of entangled measurements~\cite{Renou2018BSMSelfTest,Bancal2018BSMSelfTest}, and detection of the nonlocality of \emph{all} entangled states~\cite{Bowles2018DIAllEntState} which is out of reach in standard Bell scenarios~\cite{Tsirelson1993}. 
At a conceptual level, Network Nonlocality was recently used to propose alternative definitions for genuine multipartite entanglement~\cite{navascues2020gnme,kraft2020networkentanglement,luo2020networkentanglement} and genuine multipartite nonlocality~\cite{Coiteux2021a,Coiteux2021b}, and to understand the role of complex numbers in quantum theory~\cite{Renou2021}.

Recently, several techniques have been developed to characterize classical (also called local) and nonlocal correlations in networks \cite{wolfe2019inflation,wolfe2021qinflation,Pozas2019}. Various examples of quantum Network Nonlocal correlations are known, such as a first example in the bilocal network and its generalization~\cite{branciard2010,branciard2012,gisin2017,tavakoli2017,luo2018} or from other technics~\cite{rosset2016,tavakoli2020}.
Also, a new form of Network Nonlocality, qualified of genuine Network Nonlocality, was found in the triangle network with no input~\cite{Renou2019a}. Nevertheless, except for the last example that is generalized to ring networks with an odd number of nodes, no generic construction of network nonlocality, unrelated to the standard Bell nonlocality, is known.

This paper develops the first generic examples of Network Nonlocality in a wide class of networks with no input based on general \emph{Color-Matching} (CM) strategies, that are also presented in a parallel letter~\cite{PRL}. We also present an alternative family of strategies, called \emph{Token-Counting} (TC) that provide Network Nonlocality in another class of networks. These two techniques are unrelated to Fritz's embedding of Bell nonlocality~\cite{fritz2012}.
In a TC strategy each source distributes a fixed number of tokens and each party counts and outputs the number of received tokens. 
In a CM strategy each source takes a color and each party checks if all its neighboring sources take the same color. 
We show that in appropriate networks, these strategies produce correlations (distributions over the outputs) that are \emph{rigid} in the sense that there is essentially a unique classical strategy to generate them. 
This rigidity property, which can be thought of as the self-testing of a network classical strategy, limits the set of possible classical strategies that may generate a quantum TC or CM distribution. 
We prove rigidity of TC and CM distributions using graph theoretic tools and Finner's inequality~\cite{RenouFinner2019,Finner1992}.
We then use this limitation to prove that certain quantum TC and CM distributions in wide classes of networks are nonlocal.
This work allows to reinterpret~\cite{Renou2019a} as the preliminary fingerprint of an entirely novel approach to network nonlocality through the rigidity of some strategies in adapted classes of networks. 

In our proofs of network nonlocality we observe the emergence of a global 
{
entangled state involving all the sources and the parties of the network.
}
 For instance, in the examples based on TC, we show that in order to generate certain outputs, the tokens must be distributed according to certain patterns in the network. 
Then, in the explicit quantum strategy we observe the superposition of these patterns, which amounts to a 
{
global entangled state through the whole network. 
Since this form of entanglement is missing in the classical case, we are able to prove infeasibility of simulating those correlations by a classical strategy. 
}

The present text extends a more compact parallel letter on the same subject~\cite{PRL} and furnishes all the required technical proofs. 

\subsection{Notations} 
In this paper, we represent a network $\cN$ by a bipartite graph where one part of the vertices consists of sources and the other one consists of parties. We assume that there are $I$ sources $S_1, \dots S_I$ and  $J$ parties  $A_1, \dots,  A_J$. We write $S_i\rightarrow A_j$ and sometimes $i\to j$ when the source $S_i$ is connected to the party $A_j$. We assume that the sources are not redundant meaning that there are no distinct sources $S_i, S_{i'}$ such that $S_{i}$ is connected to all parties connected to $S_{i'}$ (as otherwise  $S_{i'}$ can be merged to $S_i$). In this paper, we consider networks with no inputs. Then, a strategy corresponds to assigning multipartite states to sources that are distributed to the parties, and assigning measurements to the parties who measure the receives subsystems. Each party $A_j$ outputs her measurement outcome $a_j$. The joint distribution of outputs is denoted by $P=\{P(a_1, \dots,  a_J)\}$.


\subsection{Classical and quantum models of network correlations}\label{sec:ClassicalQuantumModelNetworkCorrelations}

{
Given a network $\cN$ as defined above, we are interested in the set of probability distributions which can be created when the sources distribute either classical or quantum signals to the parties.
}

{
In a classical strategy, each source $S_i$ distributes classical information. Letting $\cS_i$ be the set of all potential values taken by $S_i$, the source picks a value $s_i\in \cS_i$ with some probability $p(s_i)$.
A party $A_j$ receives all source values $\{s_i:\, i\to j\}$ and provides an output $a_j$ given these values, according to a probabilistic output function $P(a_j|\{s_i:\, i\to j\})$.
Then, the joint distribution of outputs is given by
\begin{align}
P(a_1, \dots,  a_J)=
\sum_{s_1, \dots, s_I} \prod_i p(s_i)\cdot \prod_j  P(a_j&|\{s_i:i\to j\}).
\label{eq:ClassicalModel}
\end{align}
}

{
In a quantum strategy, sources distribute quantum information, and each party performs a quantum measurement. 
Without loss of generality, we assume that all states distributed by the sources are pure, and measurement performed by the parties are projective.
To be more precise, each arrow $i\to j$ of $\cN$ corresponds to a Hilbert space $\cH_{i,j}$. Source $S_i$ distributes a quantum state $\ket{\psi_i}$ in the Hilbert space $\cH_{i}=\bigotimes_{j:\,i\to j}\cH_{i,j}$, where $\cH_{i,j}$ contains the part of the state sent to party $A_j$ by source $S_i$. 
Party $A_j$ performs a projective measurement $\{\Pi_{a_j}:\, a_j\}$  consisting of projections acting on the Hilbert space $\cH_{j}=\bigotimes_{i:i\to j}\cH_{i,j}$ associated to the subsystems she receives.
Hence, the joint distribution of the outputs equals
\begin{align}
P(a_1, \dots,  a_J)=\Tr{\otimes_i\ketbra{\psi_i}{\psi_i} \cdot \otimes_j \Pi_{a_j}    }.
\label{eq:QuantumModel}
\end{align}
}

{
In this paper, we present two generic methods to obtain quantum distributions as in~\eqref{eq:QuantumModel} which do not admit any classical model as in \eqref{eq:ClassicalModel} for two large classes of networks.
}

\subsection{Overview of the paper}

In Section~\ref{sec:TC} we discuss Token-Counting (TC) strategies in networks. We first define TC strategies and distributions (Definition~\ref{definition:TC}){, and provide an overview of our method for network nonlocality from TC.}
Then, we prove our first main result, Theorem~\ref{theorem:TC}, which states that under the assumption that $\cN$ is a \emph{No Double Common-Source (NDCS) network} (see Definition~\ref{definition:NDCSnetworks}), TC distributions are rigid, meaning that only TC strategies can classically simulate distributions arising from a TC strategy.
We illustrate how network nonlocality can be obtained from TC strategies in subsection~\ref{subsec:TC5-0}.
At last, we provide an additional rigidity property of TC distributions in Subsection~\ref{subsec:ExtraRigidityTC}.

Section~\ref{sec:CM} contains similar results for Color-Matching (CM) strategies. We first define these strategies (Definition~\ref{definition:CM}){, and provide an overview of our method for network nonlocality from TC.} 
We then prove our second main result, Theorem~\ref{theorem:CM}, which shows that in \emph{Exclusive Common-Source (ECS) networks} admitting a \emph{Perfect Fractional Independent Set} (PFIS), CM strategies are rigid. We then obtain nonlocality from CM in Subsection~\ref{subsec:1-2CM} using the extra rigidity property we establish in Corollary~\ref{corollary:CMExtraMeasures}.

The presentation of general methods of TC and CM is followed by specific examples of network nonlocality in some networks. In particular, we study \emph{ring networks with bipartite sources} in Section~\ref{sec:AllRingScenariosWithBipartiteSources}, and \emph{all bipartite source complete networks} in Section~\ref{sec:BipartiteSourceCompleteNetwork}. In Section~\ref{sec:GraphColoring}, by considering a particular network, we show how proper coloring of graphs would result in network nonlocality via CM strategies.

Conclusions and final remarks are discussed in Section~\ref{sec:conclusion}. 
{
In particular, we emphasize the creation of a global entangled state through the whole network in our method, and argue that this is the fundament of the network nonlocality in our examples.
We also explain why we believe our examples of network nonlocality are essentially different from the network nonlocality obtained from existing embedding of Bell's nonlocality~\cite{fritz2012}.
}

Let us finish this subsection by advising the reader to first concentrate on the introduction, Section~\ref{sec:TC} (except Subsection~\ref{subsec:ExtraRigidityTC}) and the conclusion to obtain a good overview of the paper.

\section{Token-Counting Strategies}\label{sec:TC}

\subsection{Definitions and strategy for network nonlocality from Token-Counting}\label{subsec:DefinitionTC}


\subsubsection{Definitions}

In a \emph{Token-Counting strategy}, each source $S_i$ distributes $\eta_i$ tokens to the parties it connects. A party $A_j$ counts the number of received tokens, and may also measure extra degrees of freedom.

\begin{definition}[Token-Counting strategy]\label{definition:TC}
A strategy in $\mathcal{N}$ is called Token-Counting (TC) if, up to a relabelling of the outputs and of the information sent by the sources:
\begin{enumerate}
\item[(i)] Each source $S_i$ distributes a fixed number of tokens $\eta_i$ to the connected parties (with some fixed probability distribution) along with possibly additional information.
\item[(ii)] Every party $A_j$ counts the total number of tokens $n_j$ she receives and produces possibly additional information {$\alpha_j$. She outputs $a_j=(n_j,\alpha_j)$.}
\end{enumerate}
The resulting distribution of a \emph{Token-Counting strategy} $P=\{P((n_1,\alpha_1),\dots, (n_J,\alpha_J))\}$ is called a \emph{Token-Counting distribution}. 
\end{definition}
\begin{figure}
\includegraphics[width=0.47\textwidth]{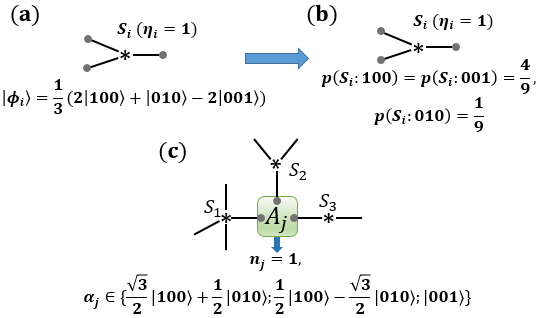}
\caption{
(a) $S_i$ distributes $\eta_i=1$ token to its adjacent three parties in superposition through the tripartite state $\ket{\phi_i}=\frac{1}{3}(2\ket{100}+\ket{010}-2\ket{001})$, a superposition of basis vectors $|\eta_i^1, \eta_i^2, \eta_i^3\rangle$ with $\eta_i^1+\eta_i^2+ \eta_i^3=\eta_i=1$. Here, e.g., $\ket{100}$
 indicates that one token is sent to the first party, and no token to the second and third parties. \\ 
(b) We restrict to the marginal over the token counts $P_{\mathrm{token}}(n_1, \dots,  n_J)$. Tokens distributed by $S_i$ are determined by measuring $\ket{\phi_i}$ in the computational basis. This reduces a quantum TC strategy to a classical one producing the same distribution $P_{\mathrm{token}}$, called the \emph{decohered classical strategy} (see Remark~\ref{remark:TC}).\\
(c) Party $A_j$ first measures the number of tokens $n_j$ she receives from all sources. Here, $n_j=1$ corresponds to the projection on the subspace spanned by vectors $\ket{001},\ket{010},\ket{100}$. Next, she measures the token provenance in a superposed way to obtain the extra output $\alpha_j$.
}
\label{fig:TC_Source_Party}
\end{figure}

We note that quantum strategies may also be TC and produce TC distributions. We may assume that each source distributes a fixed number of tokens in superposition (see Figure\,\ref{fig:TC_Source_Party}(a)). Then, each party performs a projective measurement to count the number of received tokens $n_j$. Next, she may measure other degrees of freedom to produce $\alpha_j$ (see Figure\,\ref{fig:TC_Source_Party}(c)).

Let us first remark that \emph{without} $\alpha_j$, a quantum TC distribution is always classically simulable:

\begin{remark}\label{remark:TC}
Let $P=\{P((n_1,\alpha_1),\dots, (n_J,\alpha_J))\}$ be a TC distribution associated with a quantum TC strategy in network $\mathcal{N}$. Then, the token-count marginal distribution $P_\mathrm{token}=\{P_\token(n_1,\dots, n_J)\}$ can be classically simulated by a TC strategy which we call the \emph{decohered classical strategy}. To this end, observe that the measurement operators associated with the token-counts $n_j$'s are all diagonal in the computational basis. So the sources may measure their output states in the computational basis before sending them to the parties (see Figure\,\ref{fig:TC_Source_Party}(b)). This would not change the outcome distribution of the token counts $(n_1, \dots,  n_J)$. This computational basis measurement demolishes the coherence of sources and produces the associated decohered TC classical strategy. 
\end{remark}

Unlike for a classical strategy, the measurement of the token-counts $n_j$ in a quantum strategy may leave the sources in 
{
a global superposition of several possible token distributions over the network, which could then be detected by measuring the token provenance in a superposed way with the extra information $\alpha_j$ . 
}



{
We will prove in Theorem~\ref{theorem:TC} that for some networks, TC strategies are the only classical strategies that can produce TC distributions. 
For this, we will see that we need to restrict ourselves to the following class of \emph{No Double Common-Source networks}.
}

\begin{definition}[No Double Common-Source networks]\label{definition:NDCSnetworks}
A network $\cN$ is called a \emph{No Double Common-Source (NDCS) network} if each pair of parties do not share more than one common source, i.e., there does not exist $S_i\neq S_{i'}$ and $A_j\neq A_{j'}$ such that $S_i\rightarrow A_{j}, A_{j'}$ and $S_{i'}\rightarrow A_{j}, A_{j'}$.
\end{definition}

\subsubsection{Network Nonlocality from Token-Counting}
{
In the next section, we will provide a method to obtain network nonlocality in NDCS networks via TC distributions. 
We will prove that certain quantum TC strategies produce TC distributions that cannot be simulated classically. 
By Remark~\ref{remark:TC}, it is clear that this goal can only be obtained when some extra information $\alpha_j$ is measured. 
We provide here an overview on this method which will be illustrated in the case of the network of Figure~\ref{fig:5-0}a in Section~\ref{subsec:TC5-0}. 
}

{
Our method first introduces a quantum TC strategy in an NDCS network $\cN$ achieving a distribution $P$. 
More precisely, it introduces some given number of tokens $\eta_i\geq 1$ to be distributed for each source $S_i$ of the network, as well as some concrete superposed ways to distribute them, as illustrated in Figure~\ref{fig:TC_Source_Party}a.
It also asks every party $A_j$ to first measure the total number of received tokens $n_j$, and to measure how these tokens reached her in some superposed way, as in Figure~\ref{fig:TC_Source_Party}c, obtaining a second information $\alpha_j$.
This TC strategy allows the parties to produce a TC distribution $P=\{P((n_1, \alpha_1), \dots,  (n_J, \alpha_J))\}$ corresponding to the joint probability distribution of the parties' observations. Next, assuming that the quantum TC strategy is chosen deliberately, we aim to show that $P$ cannot be produced classically, that $P$ is a nonlocal distribution in  network $\cN$.
}


{
Suppose that $P$ is local and can be simulated by a classical strategy as described in~\eqref{eq:ClassicalModel}. Then, since  $P_{\token}$ is a coarse-graining of $P$ obtained by ignoring the extra informations $\alpha_j$'s, the marginal distribution $P_{\token}$ can also be simulated classically. Indeed, as mentioned above, the corresponding decohered classical strategy (see Remark~\ref{remark:TC} and Figure~\ref{fig:TC_Source_Party}b) that is a TC classical strategy, can produce $P_{\token}$. However, the important issue is that  this TC classical strategy is essentially the unique classical strategy that produces $P_{\token}$ in the network $\cN$. This property which we call the rigidity property is proven in Theorem~\ref{theorem:TC}, and is our main step in showing Network Nonlocality via TC.
}

{
Going back to the full distribution $P$, the rigidity property substantially reduces not only the space of classical strategies that simulate the marginal $P_{\token}$, but also those of $P$ since the latter is an extension of the former. Then, further investigating the distribution $P$ and using the extra produced information $\alpha_j$'s, we will show that the whole distribution $P$ cannot be simulated by any classical strategy. This last step is discussed in details for a specific example in Subsection~\ref{subsec:TC5-0}. 
}

\subsection{Rigidity of TC strategies}\label{subsec:RigidityTC}

\begin{figure}
\includegraphics[width=0.3\textwidth]{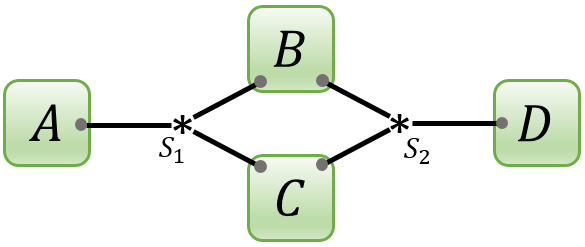}
\caption{In this network, TC distributions can be classically generated with strategies that are not TC. Suppose that $S_1$ uniformly distributes one token among $A, B, C$, while $S_2$ distributes no token. This gives a TC distribution, which however,  can alternatively be simulated with the following non-TC strategy. Suppose that $S_1$ takes value $S_1=A$ with probability $1/3$, and $S_1=BC$ with probability $2/3$. Moreover, $S_2$ takes one of the values $B$ or $C$ with uniform probability. Any of the parties outputs $n=1$ (receiving a token) if she sees her name in all the sources connected to her. This is clearly not a TC strategy, yet it simulates the initial TC distribution.}
\label{fig:TCNotCompatible}
\end{figure}

We can now state our first main result: 

\begin{theorem}[Token-Counting]\label{theorem:TC}
Let $\cN$ be a no-double-common-source network.
Then any \emph{classical} strategy that simulates a Token-Counting distribution in $\cN$ is necessarily a Token-Counting strategy. 
Moreover, in any such Token-Counting strategy sources distribute the tokens among their connected parties under a fixed (unique) probability distribution.

\end{theorem}

In Appendix~\ref{app:sec:TC} we present a more formal statement of the theorem.
{More concretely, considering an arbitrary simulating classical strategy, we show that for any $S_i\to A_j$, there is a \emph{token function} $T_i^j:s_i\in\cS_i\mapsto T_i^j(s_i)\in\mathbb Z_{\geq 0}$ mapping any value $s_i$ taken by $S_i$ to a number of tokens sent by source $S_i$ to $A_j$. We show that these token functions are consistent in the sense that:
\begin{enumerate}
\item[{\rm (i)}] $\sum_{j:i\to j} T_i^j(s_i) = \eta_i$, that is the total number of tokens distributed by $S_i$ equals $\eta_i$. 
\item[{\rm (ii)}] $n_j=\sum_{i:i\to j} T_i^j(s_i)$, that is party $A_j$ outputs $n_j$ by summing the tokens  she recieves from the adjacent sources according to the token functions. 
\item[{\rm (iii)}] The tokens are distributed with the same probabilities as in the initial TC strategy.
\end{enumerate} }

Now, we briefly explain the three ingredients used in the proof, that is also detailed in Appendix~\ref{app:sec:TC}.

\begin{proof}
First, the total number of tokens is fixed, i.e., if $P_\mathrm{token}(n_1, \dots,  n_J)>0$, then $n_1+\cdots+n_J=\eta_1+\cdots+\eta_I$. Second, by changing the output of source $S_i$, only the values of $n_j$'s with $S_i\to A_j$ may change. Third, by the NDCS assumption, for a given $S_i$ and a party $A_j$ with $S_i\to A_j$, we may fix all the messages received by $A_j$, except the one from $S_i$, to a desired value without getting a conflict with the received messages of other parties $A_{j'}\neq A_j$ with $S_i\to A_{j'}$. Then, we may study the simultaneous variations of $n_j$'s for all such parties when we change the output of $S_i$. Putting these together the proof of the theorem follows.
\end{proof}

{
Note that in networks which are not NDCS, TC distributions may be produced without TC strategies. 
For such an example  see Figure\,\ref{fig:TCNotCompatible}. 
Hence, the restriction to the class of NDCS networks is necessary for the validity of Theorem~\ref{theorem:TC}.
}


We now illustrate our method on the four-party 5-0 network consisting of five bipartite sources and zero tripartite sources. 

\subsection{Nonlocality in the 5-0 TC scenario from TC}\label{subsec:TC5-0}

\begin{figure}
\includegraphics[width=0.47\textwidth]{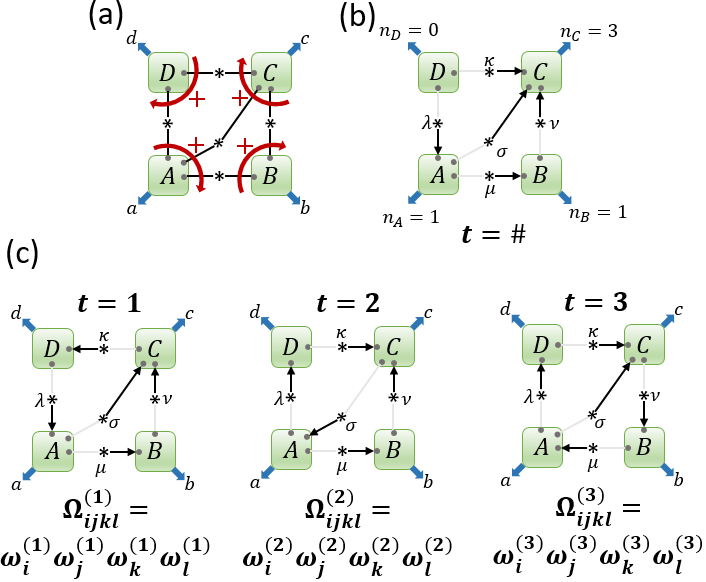}
\caption{(a) 5-0 four-party scenario and its orientation. (b) Label $t=\#$ refers to the only token distribution leading to $n_A=n_B=1, n_C=3, n_D=0$. (c) Labels $t\in\{1,2,3\}$ refer to the three possible distributions of tokens leading to $n_A=n_B=n_D=1, n_C=2$. The associated parameters $\Omega^{(t)}_{ijkl}$ are defined {in Eq.~\eqref{eq:OmegaTC} and the coefficients $\omega_i^{(1)}, ..., \omega_l^{(3)}$ are parameters of the parties' measurement bases. We assume the convention that $\omega_j^{(2)}=\omega_j^{(1)}$ and $\omega_l^{(2)}=\omega_l^{(3)}$}.}\label{fig:5-0}
\end{figure}

We demonstrate nonlocality in the four-party 5-0 network of Figure\,\ref{fig:5-0} via a TC quantum strategy. The same method will be used in Section~\ref{sec:AllRingScenariosWithBipartiteSources}. We write the outputs of the parties by $(n_A, \alpha), (n_B, \beta), (n_C, \gamma),$ and $(n_D, \delta)$, where $n_A, n_C \in \{0,1,2,3\}$ and $n_B, n_D\in\{0,1,2\}$ are the number of received tokens and $\alpha, \beta, \gamma, \delta$ are additional informations.
The TC quantum strategy is described as follows:
\begin{itemize}
\item All sources distribute the state $\ket{\psi^+}=\frac{\ket{01}+\ket{10}}{\sqrt{2}}$ meaning one token per source distributed uniformly and coherently.
\item The measurement basis vectors of $A$ are 
\begin{align*}
\quad&\ket{000},\\
\quad&\ket{\chi_{i}^{A}}=\omega_i^{(1)}\ket{100}+\omega_i^{(2)}\ket{010}+\omega_i^{(3)}\ket{001},~i= 1, 2, 3\\
\quad&\ket{110},\ket{101},\ket{011},\\
\quad& \ket{111}.
\end{align*}
In this representation of basis vectors we assume that $A$ sorts the received qubits as in Figure\,\ref{fig:5-0}(a). Observe that the first row corresponds to $n_A=0$ ($A$ receiving no token), and the other ones to $n_A=1,2,3$ respectively. We note that when $n_A\in\{0,3\}$ there is no other degree of freedom to measure. However, there are three possibilities for $n_A=1$ and $n_A=2$. In the latter case we assume that $A$ further measures in order to exactly obtain the provenance of the two tokens. In the former case, however, $A$'s measurement basis states are entangled; they correspond to projectors on $\ket{\chi_{i}^{A}}_{i=1, 2, 3}$  in which case we let $\alpha=i$. 
\item Following similar conventions as above, the measurement basis of $C$ is 
\begin{align*}
\quad&\ket{000},\\
\quad&\ket{001}, \ket{010}, \ket{100},\\
\quad&\ket{\chi_{k}^{C}}=\omega_k^{(1)}\ket{110}+\omega_k^{(2)}\ket{101}+\omega_k^{(3)}\ket{011}, ~ k=1, 2, 3\\
\quad&\ket{111}. 
\end{align*}
\item The measurement basis of $B$ is
\begin{align*}
\quad&\ket{00},\\
\quad&\ket{\chi_{j}^{B}}=\omega_j^{(1)}\ket{10}+\omega_j^{(3)}\ket{01},~j=0, 1\qquad\qquad\qquad\qquad\quad\\
\quad&\ket{11}. 
\end{align*}
\item The measurement basis of $D$ is similar to that of $B$ but with coefficients $\omega_l^{(1)}$ and $\omega_l^{(3)}$, $l=1, 2$.
\end{itemize}


{
Let $P=\{P\big((n_A, \alpha), (n_B, \beta), (n_C, \gamma),(n_D, \delta)\big)\}$ be the resulting distribution and $P_\token=\{P_\token(n_A, n_B, n_C, n_D)\}$ the corresponding token-count marginal. 
One can compute the exact probabilities applying Eq.~\eqref{eq:QuantumModel}. 
For instance, we have that $P_\token(n_A=n_B=n_D=1, n_C=2)=3/2^5$, and
\begin{align}
P\big(n_A=n_B=n_D=1, n_C=2, i, j, k, l\big)&=\nonumber\\\frac{1}{2^5}\Big|\Omega^{(1)}_{ijkl}+\Omega^{(2)}_{ijkl}+&\Omega^{(3)}_{ijkl} \Big|^2,\label{eq:OmegaTC}
\end{align}
where for $t\in\{1,2,3\}$ we define
\begin{equation}\label{eq:OmegaTC}
\Omega^{(t)}_{ijkl} = \omega^{(t)}_i\omega^{(t)}_j\omega^{(t)}_k\omega^{(t)}_l,
\end{equation}
with the convention that $\omega_j^{(2)}=\omega_j^{(1)}$ and $\omega_l^{(2)}=\omega_l^{(3)}$.
}

We aim to show that this distribution for certain choices {of the parameters $\omega_i^{(1)}, ...,\omega_l^{(3)}$} is nonlocal. To this end, assume {by contradiction} that there exists a classical strategy simulating the same distribution. Then, by Theorem~\ref{theorem:TC} and Remark~\ref{remark:TC} this strategy is a TC strategy in which each source sends one token at uniform to its connected parties. 

Let us restrict our attention to the \emph{ambiguous case}, namely when $n_A=n_B=n_D=1$ and   $n_C=2$. There are three ways of distributing the tokens in order to get $n_A=n_B=n_C=1$ and  $n_C=2$. These three cases are shown in Figure\,\ref{fig:5-0}c and are indexed by $t\in \{1, 2, 3\}$. We also introduce the label $t=\#$ in Figure\,\ref{fig:5-0}b corresponding to the unique token distribution associated to $n_A=n_B=1, n_C=3$, which is used in the proof of Claim~\ref{claim:5-0} bellow.

Let
\begin{equation*}
q(i, j, k, l, t ) = \Pr\big(i, j, k, l, t \,\big| n_A=n_B=n_D=1, n_C=2 \big),
\end{equation*}
be the probability that $\alpha=i, \beta=j, \gamma=k, \delta=l$ and $t\in \{1, 2, 3\}$ conditioned on $n_A=n_B=n_D=1$ and   $n_C=2$. By the above discussion $q(i, j, k, l, t )$ is a well-defined probability distribution. 

{
In the following, we first prove in Claim~\ref{claim:5-0} that some marginals of the distribution $q(i, j, k, l, t )$ must satisfy some equality constraints due to the network structure and the definition of $t$ as the label for the hidden way the tokens are classically distributed by the sources. 
Then, we prove in Proposition~\ref{propo:TC5-0} that for some well-chosen values of the parameters $\omega_i^{(1)}, ...,\omega_l^{(3)}$, these compatibility conditions are incompatible with the fact that $q(i, j, k, l, t )$ is a probability distribution. 
}


\begin{claim}\label{claim:5-0}
The marginals of $q(i, j, k, l, t)$ satisfy
\begin{enumerate}
\item[{\rm{(i)}}] $q(i, j, k, l)= \frac{1}{3}\big|\sum_t\Omega^{(t)}_{ijkl}\big|^2$.
\item[{\rm{(ii)}}] $q(i, t)=\frac{1}{3}|\omega_i^{(t)}|^2$, $q(j, t)=\frac{1}{3}|\omega_j^{(t)}|^2$, $q(k, t)=\frac{1}{3}|\omega_k^{(t)}|^2$ and $q(l, t)=\frac{1}{3}|\omega_l^{(t)}|^2$.
\end{enumerate}
\end{claim}
\begin{proof}
(i) This is a consequence of the definition of $q$ and the structure of the quantum strategy, giving rise to $P(n_A=n_B=n_D=1, n_C=2, \alpha=i, \beta=j, \gamma=k, \delta=l)=\frac{1}{2^5}|\sum_t \Omega_{ijkl}^{(t)}|^2$ 
{
(see Eq.~\eqref{eq:OmegaTC}).
}
\\
(ii) We only compute $q(i,t=1)$. The other cases are derived similarly. First, remark that $\Pr(\alpha=i, n_A=n_B=1,n_C=3)=\frac{1}{2^5}|\omega_i^{(1)}|^2$. Moreover, as $A$ is not connected to source $\kappa$, we have 
\begin{align*}
\Pr(\alpha=i,t=1)&=\Pr(\alpha=i,t=\#)\\
&=\Pr(\alpha=i,n_A=n_B=1,n_C=3),
\end{align*}
where we used the fact that $t=\#$ if and only if $n_A=n_B=1, n_C=3$. 
Therefore, 
\begin{align*}
q(i,t=1) &= \Pr(i, t=1)/P(n_A=n_B=n_D=1,n_C=2)\\
&= \frac{2^5}{3} \Pr(i, t=1) \\
&= \frac{2^5}{3} \Pr(\alpha=i,n_A=n_B=1,n_C=3) \\
&= \frac{1}{3}|\omega_i^{(1)}|^2.
\end{align*}
\end{proof}

Now the following proposition shows that the TC distribution cannot be simulated classically.

\begin{proposition}\label{propo:TC5-0}
For some choices of coefficients $\omega_i^{(1)}, \dots,  \omega_l^{(3)}$, no distribution $q(i, j, k, l, t)$ satisfies Claim~\ref{claim:5-0}.
\end{proposition}
\begin{proof}
Remark that the problem of finding a distribution $q(i, j, k, l, t)$ satisfying the marginal constraints given in Claim~\ref{claim:5-0}, is a Linear Program (LP). 
Solving this LP for various choices of coefficients $\omega_i^{(1)}, \dots,  \omega_l^{(3)}$, we find cases for which the LP has no solution. One can, e.g., consider identical bases for $A$ and $C$, and identical bases for $B$ and $D$, with respective bases given by the coefficients of the three-dimensional rotation matrix $R^{(3)}_x(\theta)$ of angle $\theta=\pi/8$ around the $x$-axis, and the two-dimensional rotation matrix $R^{(2)}(\theta)$ of angle $\theta$:
\begin{equation*}
R^{(3)}_x(\theta) = 
\begin{pmatrix}
1 & 0 & 0\\ 0 & \cos{\theta} & -\sin\theta \\ 0 & \sin\theta & \cos{\theta} 
\end{pmatrix}, 
R^{(2)}(\theta) = 
\begin{pmatrix}
\cos{\theta} & -\sin\theta \\ \sin\theta & \cos{\theta} 
\end{pmatrix}.
\end{equation*}

We also remark that this proposition is valid even when constraints over the coefficients are slightly different. Indeed, in the example of Subsection~\ref{subsec:1-2CM}, we may consider equal bases for $A, B, C$ and $D$, with the basis given by the coefficients of the rotation matrix $R^{(3)}_x(\theta)$.
\end{proof}

\subsection{Extra rigidity of some TC strategies}\label{subsec:ExtraRigidityTC}
Before introducing CM strategies in next section, let us discuss additional rigidity results for TC strategies.
Suppose that besides the token-count $n_j$, a party $A_j$ sometimes outputs $\alpha_j$ that exactly determines the provenance of the receives $n_j$ tokens. We claim that in NDCS networks, for a classical strategy that simulates such a TC distribution, these provenances should be output faithfully. 

More precisely, let us first introduce the notion of the rigidity of a refined measurement in a TC strategy:
\begin{definition}[Rigidity of refined measurements]
Fix a quantum TC strategy on an NDCS network $\cN$. For a party $A_j$ let $\{\eta^j_i: i\to j\}$ be an assignment of the number of tokens sent by the sources $S_i$ connected to $A_j$. We say that the computational basis state $\ket{\eta^j_i:\, i\to j}$ is \emph{rigid for $A_j$} if:
\begin{itemize}
\item[(i)] $\ket{\eta^j_i:\, i\to j}$ belongs to the measurement basis of party $A_j$
\item[(ii)] in any classical strategy that simulates the same distribution (which by Theorem~\ref{theorem:TC} is necessarily TC) $A_j$ outputs $\{\eta^j_i: i\to j\}$ if and only if the source $S_i$ with $i\to j$ sends $\eta^j_i$ tokens to $A_j$.
\end{itemize}
\end{definition}
For instance in Figure\,\ref{fig:TC_Source_Party}c, $A_j$ performs two measurements: she first applies a projective measurement to determine $n_j$, and then makes another measurement to output $\alpha_j$, the provenance of tokens. 
The measurement projection associated with $n_j=1$ equals the projection on the span of $\{\ket{100}, \ket{010}, \ket{001}\}$.  
Remark that in this example, one of $A_j$'s measurement operators corresponds to the last basis vector $\ket{001}$, i.e., to the case where the received token comes from the last source. We say that it is rigid for $A_j$ if in any classical strategy that simulates the same distribution, $A_j$ outputs $\ket{001}$ if and only if she receives exactly one token ($n_j=1$) and from its last source. In other words, $A_j$ outputs $\ket{001}$, claiming that she receives only one token and from its last source only faithfully.

The following corollary, proven in Appendix~\ref{app:sec:TC}, affirms that computational basis vectors are always rigid in NDCS networks: 


\begin{restatable}[Refined measurements in  TC]{corollary}{TCExtraMeasures}
\label{corollary:TCExtraMeasures}
Consider a quantum TC strategy in an NDCS network $\cN$. Consider a party $A_j$ and for every $S_i\to A_j$ fix a token number $\eta_i^j$. Assume that $\ket{\eta^j_i:\, i\to j}$ belongs to the measurement basis of $A_j$. Then, $\ket{\eta^j_i:\, i\to j}$ is rigid for $A_j$.
\end{restatable}

{
In the case where some refined measurements performed in the computational basis are shown to be rigid, we change the definition of $P_\token$ to include all these extra computational basis measurements.
More precisely, $P_\token$ is the distribution obtained from $P$ once one only remembers the information about the token counts and the measurements which satisfy Corollary~\ref{corollary:TCExtraMeasures}, and forgets about any extra information.
}

\section{Color-Matching Strategies}\label{sec:CM}

\subsection{Definitions and strategy for network nonlocality from Color-Matching}

\subsubsection{Definitions}

\begin{figure}
\includegraphics[width=0.47\textwidth]{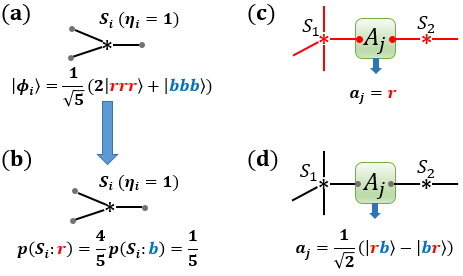}
\caption{
(a) Quantum CM for a source: the source $S_i$ distributes colors red (r) and blue (b) in the superposed way $\ket{\phi_i}=\frac{1}{\sqrt{5}}(2\ket{rrr}+\ket{bbb})$. Here, e.g., $\ket{rrr}$ indicates that the source takes color red.\\ 
(b) Corresponding decohered source, where the source colors are measured before being sent. Here $\ket{\phi_i}$ is decohered into a source distributing the red color with probability $4/5$, and the blue color with probability $1/5$.\\
(c) Color matching: $A_j$ first measures if all sources are of the same color, in which case she outputs the color value.\\ 
(d) When the colors of the sources do not match, the party $A_j$ may measure the color of sources in a superposed way, obtaining output $a_j$.}\label{fig:CMSourceParty}
\end{figure}  

In a \emph{Color-Matching strategy}, each source takes a color $c\in \{1, \dots,  C\}$ under some fixed probability distribution $p_\match(c)$. When all the colors of the sources received by a party match, she outputs the value of this matched color. Otherwise, she measures other degrees of freedom.

\begin{definition}[Color-Matching strategies]\label{definition:CM}
A strategy in a network $\cN$ is called Color-Matching (CM) if there exist a set of colors $\{1, \dots, C\}$ and a fixed probability distribution $p_\match(c)$ such that, up to a relabelling of the outputs and the information sent by the sources:
\begin{enumerate}
\item[(i)] Each source $S_i$ independently distributes a same color $c$ with probability $p_\match(c)$ to the connected parties along with possible additional information. 
\item[(ii)] Every party $A_j$ checks if all the received colors match, in which case she outputs this color. Otherwise, she possibly produces additional information.
\end{enumerate} 
The resulting distribution of a \emph{Color-Matching strategy} is called a \emph{Color-Matching distribution}.
\end{definition}
Note that contrary to TC strategies, here we assume that all the sources colors are identically distributed (yet independently), a necessary assumption for the rigidity theorem stated later.
 
Quantum strategies may also produce CM distributions. 
In this case, quantum sources distribute a superposition of different colors and $p_\match(c)$ is the distribution of the corresponding decohered  source (see Figure\,\ref{fig:CMSourceParty}a,b). 
Moreover, the measurement operators of each party include $C$ different projectors associated to matching colors $c\in \{1, \dots,  C\}$ (see Figure\,\ref{fig:CMSourceParty}c,d). The remaining measurement operators measure other degrees of freedom. 

Let us first remark that as for quantum TC strategies, a CM quantum strategy where the parties \emph{only} output color matches is classically simulable.
\begin{remark}\label{remark:CM}
Let $P=\{P(a_1,\dots, a_J)\}$ be a CM distribution associated with a quantum CM strategy in a network $\cN$. 
Consider the coarse grained distribution $P_\match=\{P_\mathrm{color}(c_1,\dots,  c_J)\}$ in which $c_j=a_j$ if $a_j$ is a matching color in $\{1,\dots,  C\}$ and $c_j=\chi$ otherwise, where $\chi$ stands for all other outputs corresponding to no color match.
Then, $P_\mathrm{color}$ can be classically simulated by a classical CM strategy which we call the \emph{decohered classical strategy}.
To this end, observe that the measurement operators associated with color matches are all diagonal in the computational basis of colors. Hence the source's colors can be measured before the state is sent to the parties, see Figure\,\ref{fig:CMSourceParty}b. 
This would not change the outcome distribution of color matches $c_1,\cdots, c_J$. This computational basis measurement demolishes the coherence of sources and produces the associated decohered classical CM strategy.
\end{remark}

As for TC, CM quantum strategies nevertheless differ from classical ones as after the quantum measurements, the sources may be in 
{
a global superposition of several possible colors distributed over the network, which could then be detected by measuring the sources colors in a superposed way. 
}

{
As for the TC case, we will prove in Theorem~\ref{theorem:CM} that for some networks, CM strategies are the only strategies that can produce CM distributions.
For this, we will need to restrict ourselves to the subclass of \emph{Exclusive Common-Source} networks admitting a \emph{Perfect Fractional Independent Set}.
}

\begin{definition}[Exclusive Common-Source network \& Perfect Fractional Independent Set]\label{definition:ECSnetworksPIS}
We say that the network $\cN$
\begin{enumerate}
\item[(i)] is an \emph{Exclusive Common-Source (ECS)} network if any source is the exclusive common-source of two parties connected to it. That is, for any source $S_i$ there exist $A_j\neq A_{j'}$ such that $S_i$ is the only source with $S_i\to A_j, A_{j'}$. 
\item[(ii)] admits a \emph{Perfect Fractional Independent Set (PFIS)} if there exists a weight $0<x_j< 1$ associated to each party $A_j$ such that for any source $S_i$, the sum of the weights $x_j$ of parties connected to $S_i$ equals 1:  
\begin{equation}
\sum_{j:i\rightarrow j} x_j = 1, \qquad \quad \forall S_i.
\end{equation} 
In other words, $\cN$ admits a \emph{Perfect Fractional Independent Set (PFIS)} if its bi-adjacency matrix $B$ admits a solution to the equation $B X=1$, with $0<X<1$.
\end{enumerate} 
\end{definition}

{
Remark that the ECS property is weaker than the NDCS assumption considered for TC strategies. 
In NDCS networks, no pair of parties can have more than one common source. This automatically yields the ECS property. 
Remark also that the PFIS assumption is readily verified for \emph{regular} networks;
a network all whose sources are $k$-partite (connected to the same number $k$ of parties), admits a PFIS  since we can simply take $x_j=1/k$ for any $A_j$.
}

\subsubsection{Network Nonlocality from Token-Counting}
{
In the next section, we will provide a method to obtain network nonlocality in ECS networks admitting a PFIS via CM distributions. 
We will prove that certain quantum CM strategies produce CM distributions that cannot be simulated classically. 
This method is similar to the TC method. }

{
Consider a given ECS network $\cN$ admitting a PFIS.
Our method first introduces a quantum CM strategy in $\cN$ achieving a distribution $P=\{P(a_1,\dots, a_J)\}$, with a list of colors and concrete superposed ways to distribute them for all sources (see Figure~\ref{fig:CMSourceParty}a) and measure them by the parties (see Figure~\ref{fig:CMSourceParty}d). The measurement operators of the parties, in particular, include projectors indicating color matches. (see Figure~\ref{fig:CMSourceParty}c).
We also consider the color-match reduced distribution $P_{\match}=\{P(c_1, \dots,  c_J)\}$, which is be obtained as a coarse-graining of $P$ by replacing all no-color-match outputs by the ambiguous output $\chi$ (see Remark~\ref{remark:CM} and Figure~\ref{fig:CMSourceParty}b).
}

{
Our method shows that for appropriate choices of the parameters of the strategy, the distribution $P$ cannot be simulated with any classical strategy.
Our proof works by contradiction.
We first assume the existence of an hypothetical classical strategy to simulate $P$.}
{
As $P_{\match}$ is a coarse-graining of $P$, this classical strategy can be used to simulate $P_{\match}$ as well.
We prove in Theorem~\ref{theorem:CM} that there is essentially a unique way to simulate $P_{\match}$, which is the decohered classical strategy described in Remark~\ref{remark:CM}.
Hence, we deduce that the hypothetical classical strategy must assign colors to sources which are used by the parties to determine their outputs.
}

{
Theorem~\ref{theorem:CM}, stating the rigidity property of CM strategies, extremely restricts the hypothetical classical strategy that simulates $P$. We remark that when no party has declared a color match, we can introduce a hidden variable $t$ corresponding to the precise way the various sources may have classically taken their colors. In the last step of our method, we show that for appropriate choices of the parameters defining the CM strategy, the existence of this additional hidden parameter $t$ is incompatible with other properties of $P$, so $P$ is nonlocal.}

{
We will illustrate this method in details  for the network of Figure~\ref{fig:1-2CM}.
}

\subsection{Rigidity of Color-Matching strategies}

We can now state our second main result.

\begin{theorem}[Color-Matching]\label{theorem:CM}
Let $\cN$ be an Exclusive Common-Source network admitting a Perfect Fractional Independent Set.
Then, any classical strategy that simulates a Color-Matching distribution in $\cN$ is necessarily a Color-Matching strategy. Moreover, in any such Color-Matching strategy sources take colors under a fixed unique probability distribution $p_\match(c)$.
\end{theorem}

{In Appendix~\ref{app:sec:CM} we present a more precise statement of the theorem.
More concretely, considering a simulating classical strategy, we show the existence of \emph{indicator color functions} $\phi_i^{(c)}:s_i\in\cS_i\mapsto \phi_i^{(c)}(s_i)\in\{0,1\}$, for every source $S_i$ and color $c$, which given value $s_i$ taken by $S_i$, indicates whether $s_i$ corresponds to color $c$. These source indicator color functions are consistent with each other, and with parties' outputs. In particular, letting $g_{j}^{(c)}$ be the characteristic function of $A_j$ outputting color match $c$, we have:
\begin{enumerate}
\item[{\rm (i)}] $g_j^{(c)}=\prod_{i:i\rightarrow j} \phi_i^{(c)}$, i.e., $A_j$ outputs color match $c$ if and only if all the sources connected to her take color $c$.
\item[{\rm (ii)}]$\forall s_i, \sum_c \phi_i^{(c)}(s_i)=1$, i.e., any value $s_i$ taken by source $S_i$ is associated to a unique color $c$.
\item[{\rm (iii)}] $\E\big[\phi_i^{(c)}\big] = p_\match(c)$, i.e., the sources take colors with the same distribution as in the original strategy.
\end{enumerate}
Note that in the short version of this work~\cite{PRL}, we adopte the different notation of a color function $c_i:s_i\in\cS_i\mapsto c_i(s_i)\in\{1,...,C\}$ mapping the value $s_i$ taken by source $S_i$ to the color $c_i(s_i)\in\{1,...,C\}$, which is the unique color with $\phi_i^{(c_i(s_i))}(s_i)=1$. We will also adopt this notation in Section~\ref{subsec:1-2CM}.
}

{
Now, we briefly explain the three ingredients used in the proof, which is detailed in Appendix~\ref{app:sec:CM}.}

\begin{proof} 
The proof relies on the use of the equality condition of Finner's inequality. 
Let $g_{j}^{(c)}$ be the characteristic function of $A_j$ outputting color match $c$ as defined above.
Since $\cN$ admits a PFIS, we may use this PFIS to write down the Finner's inequality for these characteristic functions. We observe that equality holds in this inequality so that we can impose the equality conditions of Finner's inequality. We find that for each color $c$ there is an assignment of labels ``color $c$" and ``not color $\bar{c}$" to each source. Next, we use the ECS property to show that these assignments of labels for different $c$'s match in the sense that each source takes exactly one well-defined color. Finally, the fact that the distribution of colors taken by each source is $p_\match(c)$ follows from \emph{H\"older's inequality}.
\end{proof} 

{
As the example of Figure\,\ref{fig:CMNotCompatible} shows, the ECS property is necessary for our rigidity result. 
We also need the PFIS assumption as a technical tool in the proof of our rigidity result due to our use of the equality condition of the Finner's inequality. 
However, we did not find a counter example proving it is necessary for the rigidity property.\footnote{{We were recently told that for some ECS networks not admitting a PFIS, CM rigidity still holds (private communication with Sadra Boreiri)}} 
}

\begin{figure}
\includegraphics[width=0.3\textwidth]{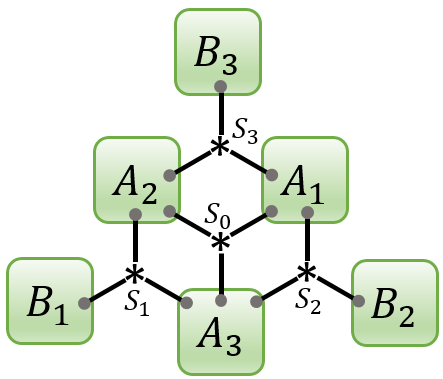}
\caption{In this network, some CM distributions can be classically simulated with strategies that are not CM. Suppose that each source distributes colors blue or red with probability $1/2$; this gives a CM distribution. The following alternative strategy simulates this distribution:\\ 
$S_1, S_2, S_3$ and $B_1, B_2, B_3$ behave similarly to the initial CM strategy. Source $S_0$  takes one of the numbers $0$ or $1$ uniformly at random. When $S_2, S_3$ have the same color $c$, party $A_1$ announces color match $c$ if $S_0=1$, otherwise she announces no match. $A_2, A_3$ behave respectively the same.\\ 
This new strategy is clearly not a CM strategy (no color can be assigned to $S_0$), yet it simulates the initial CM distribution.}
\label{fig:CMNotCompatible}
\end{figure}

\subsection{Extra rigidity of some CM strategies}\label{sec:ExtraRigidityCM}

We now discuss an additional rigidity result for CM strategies that is parallel to Corollary~\ref{corollary:TCExtraMeasures} for TC strategies. 
We will make a concrete use of this additional rigidity result in the illustration of our method in the following section.
Let us first introduce the notion of the rigidity of a refined measurement in a CM strategy:
\begin{definition}[Rigidity of refined measurements]\label{definition:CMExtraMeasures}
We say that computational basis vector $\ket{c_i:i\to j}$ is rigid for a party $A_j$ if the followings hold:
\begin{itemize}
\item[(i)] $\ket{c_i:\, i\to j}$ belongs to the measurement basis of $A_j$
\item[(ii)] in any classical strategy that simulates the same distribution (which by Theorem~\ref{theorem:CM} is necessarily CM) $A_j$ outputs $\{c_i: i\to j\}$ if and only if the source $S_i$, for any $i\to j$, takes color $c_i$.
\end{itemize}
\end{definition}
Our goal is to show that when $\ket{c_i:\, i\to j}$ belongs to the measurement basis of $A_j$, then it is rigid for $A_j$.
Nevertheless, proving such an extra rigidity result for CM strategies needs additional assumptions comparing to that for TC distributions. In the following, we fix a color $c_i$ for any source $S_i$, and assume that $\ket{c_i:\, i\to j}$ belongs to the measurement basis of \emph{any} party $A_j$. Then we prove the rigidity of $\ket{c_i:\, i\to j}$ for $A_j$. We will give similar extra rigidity results in Appendix~B, where we also give the proofs.

\begin{corollary}[Refined measurements in CM]\label{corollary:CMExtraMeasures}
Consider a quantum CM strategy in an ECS network admitting a PFIS. Fix a color $c_i$ for any source $S_i$ in the network. Assume that  for any party $A_j$, the computational basis vector $\ket{c_i:i\to j}$ belongs to the measurement basis of $A_j$. Then,  $\ket{c_i:\, i\to j}$ is rigid for $A_j$ for any $j$.
\end{corollary}

Let us illustrate the use of this corollary in the next section.
 Consider the CM strategy in the four-party 1-2 network of Figure~\ref{fig:1-2CM} with three colors. Consider, e.g., the color distribution labelled by $t=4$. Suppose that basis states $\ket{21}$, $\ket{10}$, $\ket{10}$ and $\ket{02}$ belong to the measurement bases of $A, B, C$ and $D$ respectively. Then by the above corollary, all these basis states are rigid for the associated parties. This means that in any classical strategy that simulates this CM distribution, e.g., party $A$ outputs $\ket{21}$ if and only if she receives colors $2$ and $1$ from its connected sources.

{
In the case where some refined measurements performed in the computational basis are rigid, we change the definition of $P_\match$ by including all these extra computational basis measurements.
More precisely, $P_\match$ is the distribution obtained from $P$ once one only remembers the information about the color matches and the refined measurements that are rigid, and any extra information is coarse-grained into an ambiguous output $\chi$.
} 

\subsection{1-2 CM scenario with three colors}\label{subsec:1-2CM}

We now demonstrate nonlocality in the four-party 1-2 network of Figure~\ref{fig:1-2CM} with one bipartite source and two tripartite sources. Observe that this network satisfies the ECS property and admits a PFIS (by letting $x_A=x_D=1/2$ and $x_B=x_C=1/4$). We consider a CM quantum strategy on this network in which each source distributes three colors labeled $0$=yellow, $1$=purple or $2$=red. We will use Theorem~\ref{theorem:CM} as well as Corollary~\ref{corollary:CMExtraMeasures} to prove our nonlocality result.

Our quantum CM strategy is described as follows (with the ordering of sources for each party given by orientations given in Figure\,\ref{fig:1-2CM}):
\begin{itemize}
\item The bipartite source distributes $\frac{1}{\sqrt{3}}(\ket{00}+\ket{11}+\ket{22})$, and the tripartite sources distribute $\frac{1}{\sqrt{3}}(\ket{000}+\ket{111}+\ket{222})$. That is, they distribute coherent uniform superposition of the three colors.
\item Measurement basis vectors of $A$ are:\\
$\ket{00}, \ket{11}, \ket{22}$, 
$\ket{21}, \ket{10}, \ket{02}$ and \\
$\ket{\chi_{i}^{A}}=\omega^{(1)}_i\ket{01}+\omega^{(2)}_i\ket{12}+\omega^{(3)}_i\ket{20}$ for $i=1,2,3$.

\begin{figure}
\includegraphics[width=0.47\textwidth]{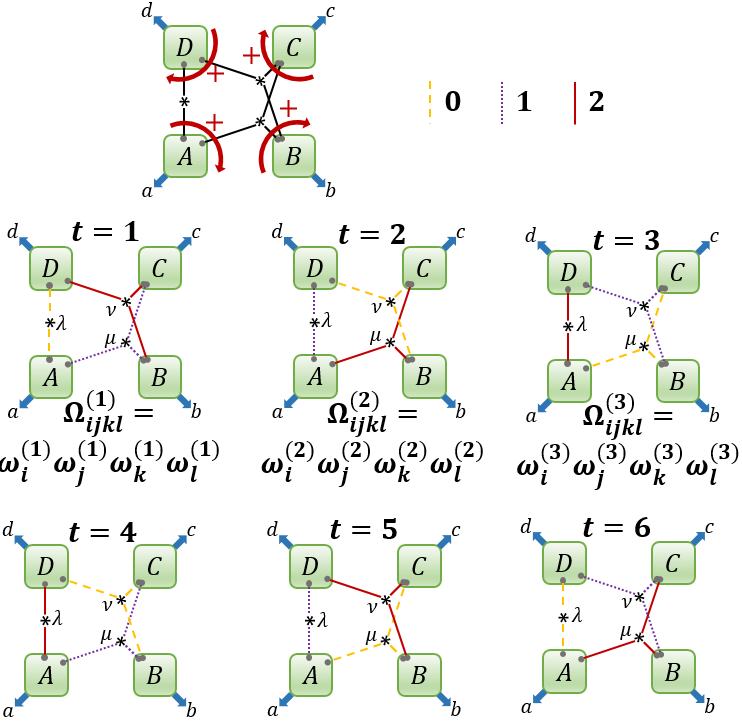}
\caption{
Four-party 1-2 network. There are six ways of choosing the color of sources when no party outputs a color match.  In instances $t=1, 2, 3$ all parties' outputs are ambiguous.
{
The associated parameters $\Omega^{(t)}_{ijkl}$ are defined accordingly in Eq.~\eqref{eq:OmegaCM} and the coefficients $\omega_i^{(1)}, \dots, \omega_l^{(3)}$ are the parameters of parties' measurement bases. As will be explained in the text, color distributions specified by $t=4, 5, 6$ correspond to refined measurements that are rigid according to Corollary~\ref{corollary:CMExtraMeasures}.
}
}\label{fig:1-2CM}
\end{figure}

Note that the first three vectors correspond to color matches. The second three ones are computational basis vectors which completely reveal the color of connected sources. The last three vectors do not reveal the color of sources and produce coherence. 

\item Measurement basis vectors of $B$ are:\\
$\ket{00}, \ket{11}, \ket{22}, \ket{10}, \ket{02}, \ket{21}$ and\\
$\ket{\chi_{j}^{B}}=\omega^{(1)}_j\ket{12}+\omega^{(2)}_j\ket{20}+\omega^{(3)}_j\ket{01}$ for $j=1,2,3$.
\item Measurement basis vectors of $C$ are:\\
$\ket{00}, \ket{11}, \ket{22}, \ket{10}, \ket{02}, \ket{21}$ and\\
$\ket{\chi_{k}^{C}}=\omega^{(1)}_k\ket{12}+\omega^{(2)}_k\ket{20}+\omega^{(3)}_k\ket{01}$ for $k=1,2,3$.
\item Measurement basis vectors of $D$ are:\\
$\ket{00}, \ket{11}, \ket{22}, \ket{02}, \ket{21}, \ket{10}$ and\\
$\ket{\chi_{l}^{D}}=\omega^{(1)}_l\ket{20}+\omega^{(2)}_l\ket{01}+\omega^{(3)}_l\ket{12}$ for $l=1,2,3$.
\end{itemize}

{
Let $P=\{P(a,b,c,d)\}$ be the resulting distribution, where $a\in\{00,11,22,21,10,02,\chi_1^A,\chi_2^A,\chi_3^A\}$, \dots,  $d\in\{00,11,22,21,10,02,\chi_1^D,\chi_2^D,\chi_3^D\}$ . It is a CM distribution since the first three measurement basis states of each party correspond to color matches. 
We also introduce $P_\match=\{P_\match(a,b,c,d)\}$ the coarse-grained CM distribution, in which all the outputs $\chi_{i}^{A}, \dots, \chi_{l}^{D}$ are replaced by a single ambiguous output $\chi$. Thus, in $P_\match=\{P_\match(a,b,c,d)\}$ the outputs satisfy $a, b, c, d\in\{00,11,22,21,10,02,\chi\}$.
We note that, e.g., $P_\match(\chi, \chi, \chi, \chi)=1/3^2$, and
\begin{align}
P\big(\chi_{i}^{A}, \chi_{j}^{B}, \chi_{k}^{C}, \chi_{l}^{D}\big)=\nonumber\frac{1}{3^3}\Big|\Omega^{(1)}_{ijkl}+\Omega^{(2)}_{ijkl}+\Omega^{(3)}_{ijkl} \Big|^2,\label{eq:OmegaCM}
\end{align}
where for $t\in\{1,2,3\}$ we use
\begin{equation}\label{eq:OmegaCM}
\Omega^{(t)}_{ijkl} = \omega^{(t)}_i\omega^{(t)}_j\omega^{(t)}_k\omega^{(t)}_l.
\end{equation}
}

{
As in the TC case of Section~\ref{subsec:TC5-0},} we assume by contradiction that this distribution is simulable by a classical strategy. 
Hence, by Theorem~\ref{theorem:CM} this strategy is a CM classical strategy. Moreover, the second triple of measurement basis states of all the parties are rigid by Corollary~\ref{corollary:CMExtraMeasures} since they correspond to color distributions $t=4, 5, 6$ in Figure\,\ref{fig:1-2CM}.

We are interested in the case where all the parties' outputs are ambiguous, i.e., their outputs are  $a=\chi_{i}^{A}, b=\chi_{j}^{B}, c=\chi_{k}^{C}, d=\chi_{l}^{D}$ for some $1\leq i,j, k, l\leq 3$. 
As mentioned above, these ambiguous cases correspond to color distributions $t=1, 2, 3$. Then, we may define 
$$q(i,j,k,l,t)=\Pr\big(i,j,k,l,t\big|\text{ambiguous outputs}\big),$$ 
be the probability distribution that $a=\chi_{i}^{A}, b=\chi_{j}^{B}, c=\chi_{k}^{C}, d=\chi_{l}^{D}$ and $t\in \{1, 2, 3\}$ conditioned on all outputs being ambiguous. By the above discussion $q(i, j, k, l, t)$ is a well-defined probability distribution.

{
In the following, we first prove that as for the TC case of Section~\ref{subsec:TC5-0}, some marginals of the distribution $q(i, j, k, l, t)$ satisfy the equality constraints of Claim~\ref{claim:5-0}. Then, we use Proposition~\ref{propo:TC5-0} to obtain a contradiction. 
}

{ 
Let us first prove that the distribution $q(i,j,k,l,t)$ satisfies Claim~\ref{claim:5-0}.
}
The marginal $q(i, j, k, l)$ is given by part (i) of this claim because as argued above, ambiguous outputs correspond to color distributions $t=1, 2, 3$. 
We prove one instance of part (ii) of the claim as other cases are similar. Using the fact that the output of $A$ is independent of the color of source $\nu$ we have
\begin{align*}
q(i&, t=1)  \\
&= \Pr(a=\chi_{i}^{A}, t=1 |\text{ambiguous outputs})\\
& = 3^2\Pr(a=\chi_{i}^{A}, t=1)\\
& = 3^2\Pr(a=\chi_{i}^{A}, c_\lambda=0, c_\mu=1, c_\nu=2)\\
& = 3 \Pr(a=\chi_{i}^{A}, c_\lambda=0, c_\mu=1)\\
& = 3^2\Pr(a=\chi_{i}^{A}, c_\lambda=0, c_\mu=1, c_\nu=0)\\
& = 3^2\Pr(a=\chi_{i}^{A}, d=0, c_\lambda=0, c_\mu=1, c_\nu=0)\\
& = 3^2\Pr(a=\chi_{i}^{A}, d=0)\\
& = \frac{1}{3}|\omega_i^{(1)}|^2.
\end{align*}

Then having the validity of Claim~\ref{claim:5-0}, we can again use Proposition~\ref{propo:TC5-0} to conclude that for certain choices of measurement parameters the resulting CM distribution is nonlocal.

\medskip
Remark that we could have chosen not to ask the parties to make measurements in the color computational basis states corresponding to color distributions $t=4, 5, 6$. For instance, $A$ could measure in a more general basis  
$\{\ket{00}, \ket{11}, \ket{22}\}\cup \{\ket{\chi_{i}^{A}}=\omega^{(1)}_i\ket{01}+\omega^{(2)}_i\ket{12}+\omega^{(3)}_i\ket{20}+\omega^{(4)}_i\ket{21}+\omega^{(5)}_i\ket{10}+\omega^{(6)}_i\ket{02}\}_{i=1,2,3,4,5,6}$, and similarly for $B$, $C$, $D$. In this case, we can again define some distribution $q(i,j, k, l, t)$, with $1\leq t\leq 6$, and compute its certain marginals. However, we find that constraints provided by these marginals do not rule out the existence of $q(i, j, k, l, t)$. Therefore, adding refined measurements and Corollary~\ref{corollary:CMExtraMeasures} crucially help in deriving nonlocality in the four-party 1-2 network via CM strategies. We will use this idea in the later examples as well.

\section{All ring scenarios with bipartite sources}\label{sec:AllRingScenariosWithBipartiteSources}

We now consider the family of networks $\cR_n$ composed of $n\geq 3$ bipartite sources $S_1, \dots,  S_n$ and $n$ parties $A_1, \dots,  A_n$ disposed in a ring. We assume that $S_i$ is connected to $A_i$ and $A_{i+1}$ where all indices here are modulo $n$ (see Figure~\ref{fig:RingScenario}). 
Remark that all $\cR_n$ are NDCS and ECS networks admitting a PFIS. Thus, both TC and CM methods can be applied. 

We consider the TC strategies where each source $S_i$ distributes $\eta_i=1$ token through the state $$\ket{\phi}=\frac{1}{\sqrt 2}(\ket{01}+\ket{10}).$$ 
Each party $A_j$ receives two qubits (one from $S_i$ and one from $S_{i-1}$) and performs projective measurement in basis $\{\ket{00}, \ket{v_{j,1}}, \ket{v_{j,2}}, \ket{11}\}$ with
\begin{equation}\label{eq:ring-v-j-r}
\ket{v_{j,r}}= \omega_{j,r}^{(1)}\ket{01}+ \omega_{j,r}^{(2)}\ket{10}, \quad r=1,2
\end{equation}
where $\omega_{j,r}^{(1)}, \omega_{j,r}^{(2)}$ are parameters to be determined. Here, we assume that the first qubit in the above equation comes from $S_{i-1}$ and the second one comes from $S_i$.
We note that $\ket{00}$ and $\ket{11}$ respectively corresponds to $n_j=0$ and $n_j=2$ received tokens, while $\ket{v_{j,1}}$ and $\ket{v_{j,2}}$ correspond to receiving $n_j=1$ token. 

Remark that for even $n$, this strategy is equivalent to a CM one. Assume that each party $A_j$ for \emph{even} $j$, flips both the received qubits in the computational basis. In this case the measurement bases would have the same structure as before, yet the distributed entangled states turn to $\frac{1}{\sqrt 2}(\ket{00} +\ket{11})$. Thus the resulting strategy is a CM strategy with two colors.

Suppose that the resulting TC distribution can be simulated by a classical strategy. By Theorem~\ref{theorem:TC} this strategy is TC: each source $S_i$ has one token and with probability $1/2$ decides to whether sends it to $A_{i}$ or $A_{i+1}$. Again we want to assume that all $A_{j}$'s outputs are ambiguous, namely, $A_j$'s output is $\ket{v_{j, r_j}}$ for some $r_j\in \{1,2\}$. Note that this happens only if any party receives exactly one token, which means that the distribution of tokens takes one of the following two forms:
\begin{itemize}
\item $t=1$: for all $i$, $S_i$ sends its token to $A_{i}$.
\item $t=2$: for all $i$, $S_i$ sends its token to $A_{i+1}$. 
\end{itemize}
For any $r_1, \dots,  r_n, t\in \{1,2\}$ define
\begin{align}
q(r_1, \dots,  r_n, t) = \Pr\big( \alpha_j=\ket{v_{j, r_j}},\,\forall j, t \,\big| \,  n_j =1\,\forall j \big ).
\label{eq:ring-def-q}
\end{align}
By the above discussion, $q(r_1, \dots,  r_n, t)$ is a well-defined probability distribution.

\begin{claim}\label{claim:RingScenario}
The marginals of $q$ satisfy:
\begin{enumerate}
\item[{\rm{(i)}}] $q(r_1, \dots,  r_n)=\frac{1}{2} \Big|\prod_{j} \omega^{(1)}_{j, r_j} + \prod_{j} \omega^{(2)}_{j, r_j}  \Big|^2$
\item[{\rm{(ii)}}] $q(r_j, t) =\frac{1}{2} |\omega_{j, r_j}^{(t)}|^2, \, \forall j$. 
\end{enumerate}
\end{claim}

\begin{figure}
\includegraphics[width=0.4\textwidth]{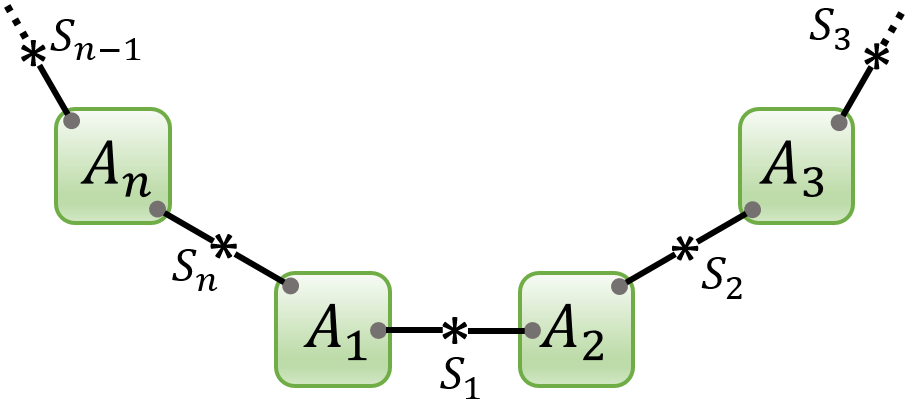}
\caption{The ring scenario $\cR_n$}\label{fig:RingScenario}
\end{figure}

As shown in Appendix~\ref{app:sec:ring}, the proof of this claim is based on similar ideas to that of Claim~\ref{claim:5-0}.

\begin{proposition}\label{propo:RingScenario}
For some choices of coefficients $\omega_{j, r_j}^{(t)}$, no distribution $q(r_1, \dots,  r_n, t)$ satisfies Claim~\ref{claim:RingScenario}
\end{proposition}

The proof of this proposition is left for Appendix~\ref{app:sec:ring}. This completes the proof of nonlocality in rings. 

\medskip
Before moving to the next set of examples let us remark that with similar ideas as in the proof of Claim~\ref{claim:RingScenario}, we can show that, e.g., 
$$q(r_1, \dots,  r_{n-1}, t) = \frac{1}{2}\Big|\prod_{j=1}^{n-2} \omega_{j, r_j}^{(t)} \Big|^2.$$
Such equations give stronger constraints on coefficients $\omega_{j, r_j}^{(t)}$ comparing to the ones of Claim~\ref{claim:RingScenario}. These constraints are not needed to prove Proposition~\ref{propo:RingScenario}. However, we may use them to obtain qualitative improvements, e.g., to reject more $\omega_{j, r_j}^{(t)}$'s.

Also remark that the ring example may be embedded in any NDCS network with a loop. By assuming that some of the sources do not have any token ($\eta_i=0$), we can essentially ignore them. Then, as before we can associate one token to any source on the ring and repeat the same calculations to obtain nonlocality in such a network. 

\section{Bipartite-sources Complete Networks}\label{sec:BipartiteSourceCompleteNetwork}

Let $\cK_n$ be the network with $n$ parties $A_1, \dots,  A_n$ and $\binom{n}{2}$ sources $S_{jj'}$ for any pair $1\leq j< j'\leq n$. We assume that $S_{jj'}$ is bipartite and is connected  to $A_j$ and $A_{j'}$.
In the following we prove nonlocality in network $\cK_n$ for any $n> 3$, using a CM strategy with two colors.
To this end we use Theorem~\ref{theorem:CM} as well as an extension of Corollary~\ref{corollary:CMExtraMeasures} which lead us to a distribution satisfying the same constrains as in Claim~\ref{claim:RingScenario}. Finally, we use Proposition~\ref{propo:RingScenario} to prove nonlocality.

First, we describe our quantum CM strategy. We assume that the number of colors is $C=2$, and each source distributes the maximally entangled state $\frac{1}{\sqrt 2}\big(\ket{00}+ \ket {11}\big)$. Thus, party $A_j$ receives $(n-1)$ qubits. We assume that she measures the received qubits in the orthonormal basis $\cB\cup \{\ket{v_{j,1}}, \ket{v_{j,2}}\}$ where
\begin{align}
\cB= \big\{\ket{x}:\, x\in \{0,1\}^n \big\}\setminus \{  \ket{01\cdots 10}, \ket{10\cdots 01}    \},
\end{align}
and 
\begin{align}\label{eq:Kn-v}
\ket{v_{j,r}} = \omega_{j,r}^{(1)}\ket{01\cdots 10} + \omega_{j,r}^{(2)}\ket{10\cdots 01}, ~~ r=1,2,
\end{align}
for parameters $\omega_{j,r}^{(1)}, \omega_{j,r}^{(2)}$ to be determined. 
Hence, all parties always measure in the color computational basis, except for a two-dimensional subspace spanned by $\ket{v_{j,1}},\ket{v_{j,2}}$. Here, we assume that $A_j$ orders the received qubits according to the orientation given in Figure~\ref{fig:Kn}.

\begin{figure}
\includegraphics[width=0.43\textwidth]{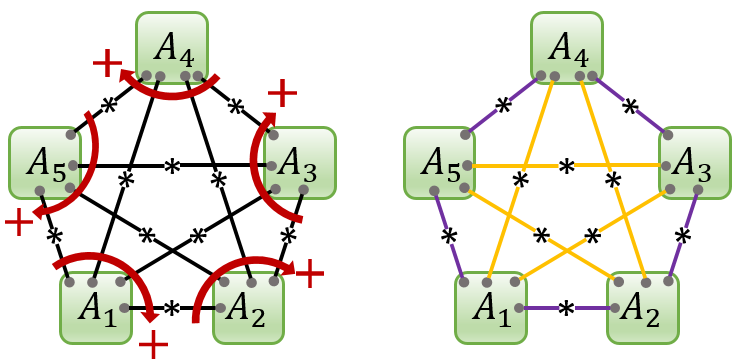}
\caption{Bipartite-sources Complete Network. Each node corresponds to a party and each edge represents a source. Each party sorts the received qubits in the clockwise order. If any $A_j$ outputs $\ket{v_{j, r_j}}$ for some $r_j$, then the sources must take colors as in the picture, i.e., all the \emph{internal} sources must take the same color, different from the $n$ \emph{boundary} sources. }
\label{fig:Kn}
\end{figure}

We note that as $\ket{0\cdots 0}$ and $\ket{1\cdots 1}$ belong to the measurement bases, the quantum strategy is CM. Moreover, the network satisfies ECS and admits a PFIS. Thus if a classical strategy simulates the same outcome distribution, by Theorem~\ref{theorem:CM} it is necessarily CM. 
On the other hand, by an extension of Corollary~\ref{corollary:CMExtraMeasures} given in Appendix~B (see part (i) of Corollary~\ref{corollary:CMExtraMeasures}) we find that all the computational basis measurement states of the parties are rigid. Thus restricting to the case where all parties' outputs are ambiguous, there remains  only two patterns of distributions of colors:

\medskip
\begin{minipage}[h]{3in} 

\emph{For $t=1,2$ all sources $S_{j,j+1}$ take color $t$, and the other sources $S_{j,j'}, |j'-j|>1$ take color $1-t$ (see Figure~\ref{fig:Kn}). }
\end{minipage}

\medskip
Then as before
we may define $q(r_1, \dots,  r_n, t)$ to be given by
\begin{align*}\Pr\big
(A_j = \ket{v_{j, r_j}}\,\forall j,\, t \big|   \text{ ambiguous outputs} \big).
\end{align*}
By the above discussion, $q(r_1, \dots,  r_n, t)$ is a well-defined probability distribution. Moreover, as shown in Appandix~D, this distribution satisfies Claim~\ref{claim:RingScenario}. Hence, by Proposition~\ref{propo:RingScenario}, there are constants $\omega_{j, r}^{(t)}$ for which there is no distribution $q(r_1, \dots,  r_n, t)$ satisfying Claim~\ref{claim:RingScenario}. Therefore, the given CM distribution is nonlocal.

\section{Color-Matching scenario via graph coloring}\label{sec:GraphColoring}

In this section we explain how the combinatorial problem of proper graph coloring can be used to construct examples of network nonlocality. To illustrate our ideas we use the complete graph, but the ideas work essentially for any graph.

We start by the description of the network $\cN$. Suppose we have $n$ sources $S_1, \dots, S_n$ and for any $1\leq i<j\leq n$ we have party $A_{ij}$ that is connected to the sources $S_i, S_j$. This corresponds to a complete graph with $n$ vertices (the sources) and $\binom{n}{2}$ edge between each pair of source (the parties). Observe that this network satisfies ECS (as $A_{ij}$ is connected to only $S_i$ and $S_j$) and admits a PFIS (since the associated graph is regular).

Consider a CM strategy on this network with $C=n$ colors. We would like to assume that no party outputs a color match. To this end, letting $c_i$ be the color taken by sources $S_i$, we need that $c_i\neq c_j$ for any $i\neq j$ as otherwise the party $A_{ij}$ outputs a color match. Then, we obtain a proper coloring of the complete graph. This means that color distributions associated with the interesting case where all outputs are ambiguous correspond to proper colorings. We note that there are $n!$ proper colorings of the complete graph, one for each permutation of the colors. 

As before, to reduce the above number we add refined measurements in the computational basis and use an extension of Corollary~\ref{corollary:CMExtraMeasures}. With this idea we reduce the number of color distributions resulting in ambiguous outputs to two ($t=1,2$). Next, we define some distribution $q$ for which we verify the validity of Claim~\ref{claim:RingScenario}. Finally, using Proposition~\ref{propo:RingScenario} we conclude that for certain choices of measurement parameters, the resulting CM distribution is nonlocal. We leave the details of this argument for Appendix~\ref{app:sec:GraphColoring}.

Here, we would like to emphasize that the idea behind this example is quite general and works for a large class of graphs. Starting with an arbitrary graph, we may think of its vertices as sources, and its edges are parties. Then, take a CM strategy with $C$ colors, where $C$ is the coloring number of the graph. This network satisfies ECS, and assuming that it admits a PFIS (which holds if the graph is regular), we can apply Theorem~\ref{theorem:CM}. Corollary~\ref{corollary:CMExtraMeasures} may also be used to simplify the study the resulting distribution and proving nonlocality.

\section{Conclusion and final remarks}\label{sec:conclusion}

In this paper we proposed two general methods for deriving nonlocality in wide classes of networks. Our methods are based on the crucial observation that Token-Counting and Color-Matching distributions are rigid. That is, in order to classically simulate such distributions we are forced (in certain networks) to use TC and CM strategies. These rigidity properties substantially restrict the set of potential classical strategies that can simulate such distributions. Then, further study of these strategies leads us to examples of  nonlocality in networks. 

\subsection{Superposition of tokens and colors}

{
We would like to emphasize that our examples of network nonlocality are fundamentally different from the existing embedding of Bell's nonlocality into network scenarios. 
As argued by Fritz~\cite{fritz2012}, one may embed nonlocal distributions of standard Bell's scenarios in networks. 
Below, we briefly explain this construction.
}

{
Suppose that the parties, Alice and Bob, share a two-qubit Bell state, with the goal of violating the CHSH Bell inequality.
Testing the CHSH inequality, unlike the network scenarios considered in this paper, requires local inputs for both Alice and Bob. Nevertheless, this CHSH test can be embedded in the triangle network with no inputs and two bits of outputs per party. To this end, the effective inputs of the CHSH test are provided by the two additional sources: the source shared between Alice and Charlie and the one shared between Bob and Charlie; any of these sources provides a uniform random bit shared between the corresponding parties. These two random bits are used by Alice and Bob as their inputs for the CHSH test.
We assume that all parties output the received ``input bits.''  
The correspondence between these output bits ensures that Alice and Bob output exactly the bits that they share with Charlie.
Finally, Alice and Bob both additionally output the the measurement outcomes of the CHSH test performed on the shared Bell state.
Fritz~\cite{fritz2012} showed that if the resulting distribution can be reproduced by a classical strategy in the triangle network, then the distribution can violate the CHSH inequality, which is impossible. Thus, the quantum strategy provides an example of Network Nonlocality.
Of course, this embedding can be generalized for a large class of networks to obtain other examples of Network Nonlocality. 
Nevertheless, in such examples several sources and parties of the network only have a classical behavior.
}

{
We believe that our examples of Network Nonlocality are fundamentally different from the above construction.
Let us discuss this difference via the TC example of the ring network given in Section~\ref{sec:AllRingScenariosWithBipartiteSources}. 
In this example, after using the rigidity of TC strategies, we considered the case where all the parties' outputs are in the ambiguous case. That is, we assumed that each party receives one token with the ambiguity being in its provenance. We observed that in this case, in any simulating classical strategy all sources must distribute their tokens either in the clockwise or in the counter-clockwise directions, respectively denoted $t=\circlearrowright$ and $t=\circlearrowleft$ here. Remark that the same holds in the initial explicit quantum strategy we considered, where the tokens are now in a \emph{superposition} of those two directions. Indeed, when all the parties project on the subspace of receiving exactly one token, the global entangled state shared between them is proportional to 
\begin{equation}\label{eq:GlobalEntangled}
\ket{\circlearrowright} + \ket{\circlearrowleft}.
\end{equation}
In the classical case, however, $t$ must be a hidden variable that takes one of the values $t=\circlearrowright$ or $t=\circlearrowleft\}$. This is why we introduced the joint distribution $q$ (see Eq.~\eqref{eq:ring-def-q}), aiming to simulate this coherent superposition in a classical incoherent way. Not surprisingly, we demonstrated that the joint distribution $q$, including the hidden variable $t$, cannot exist for appropriate choices of the measurement parameters, and proved Network Nonlocality. 
The same discussion adapts to all our examples in which a more general entangled state $\sum_t\ket{t}$ is created.
}

{
To summarize, the main feature of all our examples of Network Nonlocality is the creation of a global entangled state involving all sources and parties of the networks. This feature is not present in examples of Network Nonlocality via standard Bell's scenarios, in which no  such global entangled state is created.
Note, however, that we do not prove the necessity of the creation of the global entangled state of Eq.~\eqref{eq:GlobalEntangled}. 
Indeed, our proof  does not exclude the possibility of generating the same nonlocal distribution with another \emph{quantum} strategy in which this global entangled state is not present. 
We leave this as an open question for future works. An approach is answer this question could be to find a self-testing proof that shows that the quantum states and measurements used in our protocols are essentially the unique states and measurements that yield the target probability distribution~$P$.
}

We also remark that for some networks such a global coherent state cannot be created.
For instance, let us consider a TC strategy in a network $\cN$ in which the removal of a source $S_i$ creates two disjoint components $\cN_1, \cN_2$. In this case, the total number of tokens sent by $S_i$ to the parties in $\cN_1$ can easily be deduced by looking at the total number of tokens measured by the parties in $\cN_2$.
This property makes the creation of
{
a \emph{global coherent state} similar to the one of~\eqref{eq:GlobalEntangled} in
 $\cN$ via TC impossible.}

\subsection{Experimental realisations?}


{
Let us now discuss potential noisy experimental implementations. 
Our proofs in the TC case are based on graph theoretical and combinatorial tools which we do not know how can be modified in the presence of noise. 
In the CM case, however, we use analytic tools (Finner's inequality) which can be adapted to the noise tolerant regime, in which the states distributed by the sources may not be pure. 
To this end, the ideas in~\cite{EFKY16} on the stability of the \emph{Loomis-Whitney inequality} (that is a special case, yet essentially equivalent version, of Finner's inequality) can be used to prove a bound on the noise tolerance of our examples of Network Nonlocality via CM strategies. 
Indeed, in the proof of Theorem~\ref{theorem:CM} we use the equality condition of Finner's inequality for certain functions associated to a color to establish the existence of proper color functions. 
Now to prove noise tolerance, the equality condition is replaced with an \emph{almost} equality condition. 
Applying Corollary~\ref{corollary:CMExtraMeasures} of~\cite{EFKY16} directly proves the existence of \emph{approximate} color functions. 
This is the main step in the proof of noise tolerance, showing a noise tolerant version of the property $\mathrm{(i)}$ bellow Theorem~\ref{theorem:CM}.
From this, noisy versions of properties $\mathrm{(ii)}$ and $\mathrm{(iii)}$ can also be deduced, which results in a noise tolerant version of Theorem~\ref{theorem:CM}.
This theorem can then be used to show (in a quantifiable way) that even a noisy quantum CM distribution cannot be simulated classically and is nonlocal. Indeed, once we established a noise tolerant version of Theorem~\ref{theorem:CM}, we can take the same approach as before, and use Claims~\ref{claim:5-0},~\ref{claim:RingScenario} and~\ref{claim:CMKn} as well as Propositions~\ref{propo:TC5-0} and~\ref{propo:RingScenario} (which are noise tolerant) to show that for sufficiently low noise level, the noisy CM distribution is nonlocal.
}

{
Although the above approach does give a noise tolerant rigidity for CM distributions, this direct adaptation of our proof results in an extremely weak (experimentally not realistic) noise tolerance.
It would be desirable to find new proof techniques for the rigidity of CM distributions that are well-adapted in the noisy regime.
}

{
Alternatively, one may consider optimization approaches to estimate this noise tolerance. 
For instance, the recent machine learning algorithms developed in~\cite{Krivvachy2020} already predict an experimentally reasonable noise tolerance for the network nonlocal distribution of~\cite{Renou2019a}, and could directly be adapted to all our examples (for networks of small sizes). 
Note that the noise tolerance deduced from these adapted optimization algorithms would not be a rigorous noise tolerance value, yet it can be used as a benchmark for experiments.
}

\subsection{Conclusion}

Finally, we contemplate our two examples of TC and CM strategies as the first examples of a potential general method to derive Network Nonlocality based on combinatorial primitives. We discussed that in our TC example in the ring network, the creation of the superposition of two orientations, \emph{associated with giving a direction to each edge} in the ring graph, is the origin of network nonlocality. Moreover, we observed that nonlocality in the example of Section~\ref{sec:GraphColoring} is emerged from the coherent superposition of proper colorings of the complete graph. Orientations of the edges of a graph, and proper colorings of a graph may be the first examples of a general method based on combinatorial primitives in networks, whose coherent superposition leads to Network Nonlocality. 

\medskip
\textit{Acknowledgements.}
We thank Antonio Ac\'in and Nicolas Gisin for discussions.
M.-O.R. is supported by the Swiss National Fund Early Mobility Grants P2GEP2\_19144 and the grant PCI2021-122022-2B financed by MCIN/AEI/10.13039/501100011033 and by the European Union NextGenerationEU/PRTR, and acknowledges the Government of Spain (FIS2020-TRANQI and Severo Ochoa CEX2019-000910-S [MCIN/ AEI/10.13039/501100011033]), Fundació Cellex, Fundació Mir-Puig, Generalitat de Catalunya (CERCA, AGAUR SGR 1381) and the ERC AdG CERQUTE.

\bibliographystyle{apsrev4-2}
\bibliography{references}

\begin{appendix}

\section{Rigidity of TC distributions}\label{app:sec:TC}

Recall that a network consists of sources $S_1, \dots,  S_{I}$ and parties $A_1, \dots,  A_J$ in which $S_i\to A_j$ (or $i\to j$ when there is no confusion) means that the source $S_i$ is connected to the party $A_j$. 
In a classical strategy for such a network each source $S_i$ takes a value $s_i\in \mathcal{S}_i$ with some fixed distribution over $\mathcal S_i$ and sends $s_i$ to its connected parties. Then party $A_j$ computes a function $a_j=a_j(\{s_j: ~ S_j\to A_i\})$ of all received messages as her output. 
{
In a TC distribution the output of $A_j$ is a pair $a_j=(n_j,\alpha_j)$: the token count number denoted by $n_j$ and the other part denoted by $\alpha_j$, both being functions of $\{s_j: ~ S_j\to A_i\}$.
} 

Here, we first prove Theorem~\ref{theorem:TC} of the main text which we rephrase for convenience. 

\begin{theorem}\label{app:theorem:TC}
Let $\cN$ be a No Double Common-Source network with parties $A_1, \dots,  A_{J}$ and sources $S_1, \dots,  S_{I}$. Fix a strategy in which the source $S_i$ distributes $\eta_i$ tokens and let $P=\{P((n_1, \alpha_1), \dots,  (n_J, \alpha_J))\}$ be the corresponding TC distribution over $\cN$. For any possible token distribution $\{t_{i}^j: S_i\to A_j\}$ with $\sum_{j: S_i\to A_j} t_i^j=\eta_i$, let $q_i(\{t_{i}^j: S_i\to A_j\})$ be the probability that in this strategy $S_i$ distributes $\{t_i^j\}_{j:i \to j}$ tokens to the set of parties $\{A_j\}_{j:i \to j}$ connected to it. 

Now consider another strategy that simulates $P=\{P((n_1, \alpha_1), \dots,  (n_J, \alpha_J))\}$ on $\cN$. Then, this strategy is a TC strategy with the same distribution of tokens as before. More precisely, for any strategy that simulates $P$ there are functions $T_i^j: \mathcal S_i\to \mathbb Z_{\geq 0}$ for any $S_i\to A_j$ such that 
\begin{enumerate}
\item[{\rm (i)}] $\sum_{j: S_i\to A_j} T_i^j(s_i)=\eta_i$ for all $S_i$ and $s_i\in \mathcal S_i$.
\item[{\rm (ii)}] $n_j(\{s_i:~S_i\rightarrow A_j\}) = \sum_{i:S_i\rightarrow A_j} T_i^j(s_i)$. 
\item[{\rm (iii)}] For any source $S_i$ and any $\{t_{i}^j:\, S_i\to A_j \}$ with $\sum_{j: S_i\to A_j} t_i^j=\eta_j$ we have 
$$\Pr_{s_i}[T_i^j(s_i)=t_i^j,\, \forall j:\, S_i\to A_j] = q_i(\{t_{i}^j: S_i\to A_j\}).$$
\end{enumerate}
\end{theorem}

This theorem says that in any strategy that simulates $P=\{P((n_1, \alpha_1), \dots,  (n_J, \alpha_J))\}$, any symbol $s_i$ distributed by $S_i$ corresponds to sending $T_i^j(s_i)$ tokens to $A_j$ if $S_i\to A_j$. By (i) the total number of distributed tokens by $S_i$ equals $\eta_i$. Next, by (ii) each party to generate her first part of the output simply counts the number of received tokens. This mean that it is a TC strategy. Finally (iii) says that any TC strategy that simulates $P$ must distribute tokens with the same probability distribution as in the original strategy.

\begin{proof}

Fix some $r_i\in \cS_i$ for any source $S_i$. For any $S_i\to A_j$ define
\begin{align*}
R_i^j(\{s_{i'}:\,S_{i'}\rightarrow A_j\}) = n_j(\{s_{i'}&:\,S_{i'}\rightarrow A_j\}) \\
&- n_j(\{\hat s_{i'}^i:\, S_{i'}\rightarrow A_j\}),
\end{align*}
where 
\begin{align*}
\hat s_{i'}^i = \begin{cases}
s_{i'} \quad i'\neq i,\\
r_i \quad i'=i.
\end{cases}
\end{align*}
We note that $R_i^j$ computes the difference of the token count number of party $A_j$ when her message from $S_i$ is changed from $s_i$ to $r_i$ while other messages remain the same. In some sense, $R_i^j$ is the derivative of $n_j$ with respect to the $i$-th message.  

Observe that by changing $s_i$ to $r_i$ while leaving the other messages the same, only the outputs of parties connected to $S_i$ may change. Moreover, as a TC distribution, if $P((n_1, \alpha_1), \dots,  (n_J, \alpha_J))>0$, the total number of tokens $\sum_j n_j = \sum_i \eta_i$ is fixed independent of messages. Therefore, we have
\begin{align}\label{app:eq:t-conservation}
\sum_{j: S_i\to A_j} R_i^j(\{s_{i'}:\,S_{i'}\rightarrow A_j\}) =0, \qquad \forall S_i.
\end{align}

Let $S_{i'}\neq S_i$. Recall that by assumption $S_{i'}$ cannot share more than one connected party with $S_i$. This means that $s_{i'}$ appears  \emph{at most once} in the right hand side of~\eqref{app:eq:t-conservation}. Therefore, since the left hand side is a constant, all the terms are independent of $s_{i'}$. This means that 
$$R_i^j(\{s_{i'}:\,S_{i'}\rightarrow A_j\}) = R_i^j(s_{i}),$$
is a function of $s_i$ only and is independent of other arguments. 

Next, using the definition of $R_i^j$ we have
$$n_j(\{s_{i'}:S_{i'}\rightarrow A_j\}) =
 R_i^j(s_{i})+n_j(\{\hat s_{i'}^i: S_{i'}\rightarrow A_j\}).$$
Writing down the same equation for $n_j(\{\hat s_{i'}^i: S_{i'}\rightarrow A_j\})$ with respect to another source, and replacing $s_{i'}$'s with $r_{i'}$'s one by one, we find that 
\begin{align}\label{app:eq:R-i-j}
n_j(\{s_{i}:S_{i}\rightarrow A_j\}) = \sum_{i:S_{i}\to A_j} R_i^j(s_i) + n_j(\{r_{i}:S_{i}\rightarrow A_j\}).
\end{align}

For any source $S_i$ and party $A_j$ with $S_i\to A_j$ let 
$$\ell_i^j = \min\{\eta_i^j:\, q_i(\eta_i^j)>0\},$$
Then, we have 
$$ n_j^{\min} = \sum_{i: S_i\to A_j} \ell_i^j.$$
where $n_j^{\min} = \min\{n_j:\, P(n_j)>0\}$. Then, taking the minimum of both sides in~\eqref{app:eq:R-i-j} we find that 
$$\sum_{i:S_{i}\to A_j}  \min_{s_i} R_i^j(s_i) + n_j(\{r_{i}:S_{i}\rightarrow A_j\})=n_j^{\min}.$$
Therefore, letting 
$$T_i^j(s_i) = R_i^j(s_i) - \min_{s'_i}  R_i^j(s'_i) +\ell_i^j,$$
we fine that 
$$n_j(\{s_{i}:S_{i}\rightarrow A_j\}) = \sum_{i:S_{i}\to A_j} T_i^j(s_i).
$$
We also note that by definition, $R_i^j$ and then $T_i^j$ take integer values and we have $T_i^j(s_i) = R_i^j(s_i) - \min_{s'_i}  R_i^j(s'_i) +\ell_i^j\geq \ell_i^j\geq 0$. These give 
(ii).

We now prove (i). Fix a source $S_i$. We compute
\begin{align*}
\sum_{j: S_i\to A_j}& n_j(\{s_{i'}:S_{i'}\rightarrow A_j\})  \\
&= \sum_{j: S_i\to A_j} \sum_{i':S_{i'}\to A_j} T_{i'}^j(s_{i'})\\
&=\sum_{j: S_i\to A_j} T_i^j(s_i) + \sum_{j: S_i\to A_j} \sum_{i'\neq i:S_{i'}\to A_j} T_{i'}^j(s_{i'}).
\end{align*}
Take the minimum of both sides over all $s_{i'}$'s with $i'\neq i$. Since $P(n_j=0)>0$, we know that $\min_{s_{i'}} T_{i'}^j(s_{i'})=0$. Moreover, by the NDCS assumption, any $s_{i'}$ for $i'\neq i$ appears only once in the right hand side. Therefore, we have
\begin{align*}
\min_{s_{i'}: i'\neq i}\sum_{j: S_i\to A_j}& n_j(\{s_{i'}:S_{i'}\rightarrow A_j\})  &=\sum_{j: S_i\to A_j} T_i^j(s_i). 
\end{align*}
Observe that as a token counting distribution in which $S_i$ distributes $\eta_i$ tokens, the left hand side is at least $\eta_i$. Therefore, 
\begin{align}\label{app:eq:eta-i-ineq}
\eta_i\leq \sum_{j: S_i\to A_j} T_i^j(s_i). 
\end{align}
Summing the above inequality over all sources $S_i$ and rearranging the sum, we find that
\begin{align*}
\sum_i \eta_i\leq \sum_i \sum_{j: S_i\to A_j} T_i^j(s_i)=  \sum_{j} n_j(\{s_{i}:S_{i}\rightarrow A_j\}). 
\end{align*}
We note that the right hand side is the total number of tokens. Thus, we have equality here and in~\eqref{app:eq:eta-i-ineq} for any $i$. This gives (i).

Part (iii) is proven in Lemma~\ref{app:lem:token-dist-match} below.

\end{proof}

\begin{lemma}\label{app:lem:token-dist-match}
Let $\cN$ be a No Double Common-Source network with parties $A_1, \dots,  A_{J}$ and sources $S_1, \dots,  S_{I}$. Let $P=\{P((n_1, \alpha_1), \dots,  (n_J, \alpha_J))\}$ be a TC distribution over $\cN$ in which the source $S_i$ distributes $\eta_i$ tokens. Consider two TC strategies for simulating $P$ on $\cN$ that satisfy parts (i) and (ii) of Theorem~\ref{app:theorem:TC}. More precisely, we assume that there are sets $\mathcal S_i^{(u)}$, for $u=1, 2$, and functions $T_i^{(u),j}: \mathcal S_i^{(u)}\to \mathbb Z_{\geq 0}$ for any $S_i\to A_j$ such that (i) and (ii) hold and for any $(n_1, \dots,  n_J)$ we have
$$\Pr\Big[\sum_{i: S_i\to A_j} T_i^{(u), j}(s_i^{(u)}) = n_j,~ \forall j\Big] = P_\token(n_1, \dots,  n_J). $$
Then, for any $S_i$ and $\{t_i^j:\, S_i\to A_j\}$ we have
\begin{align*}
\Pr\big[T_i^{(1),j}(s_i^{(1)})=&t_i^j,\, \forall j:\, S_i\to A_j\big]\\ 
&= \Pr\big[T_i^{(2),j}(s_i^{(2)})=t_i^j,\, \forall j:\, S_i\to A_j\big].
\end{align*}

\end{lemma}

\begin{proof}
We prove the lemma by induction on $J$, the number of parties.  Observe that if $I=1$, i.e., there is a single source, the marginal distribution of outputs over the tokens $P_\token(n_1, \dots,  n_J)$ equals $\Pr(T_1^{(u),1}=n_1, \dots,  T_1^{(u),J}=n_J)$ in which case there is nothing to prove. Thus, we assume that there are at least two sources. 

Let $S_i$ be an arbitrary source and let $A_{j_0}$ be a party \emph{not} connected to it. (Note that if all parties are connected to $S_i$, by the NDCS assumption $S_i$ would be the unique source.) 
Let 
$$n_{j_0}^{\min} = \min\{n_{j_0}:\, P(n_{j_0})>0\},$$
be the minimum number of tokens that can be sent to $A_{j_0}$.
For any $S_{i'}$ with $S_{i'}\to A_{j_0}$ let 
$$\ell_{i'}^{(u), {j_0}} = \min_{s_{i'}^{(u)}} T_{i'}^{(u),{j_0}}(s_{i'}^{(u)}).$$
Then, $n_{j_0}^{\min} = \sum_{i': S_{i'}\to A_{j_0}} \ell_{i'}^{(u),{j_0}}$ and we have 
\begin{align}\label{app:eq:prop-n-j-max}
P_\token(n_{j_0}=n_{j_0}^{\min}) & = \prod_{i': S_{i'}^{(u)}\to A_{j_0}}  \Pr[\,T_{i'}^{(u),{j_0}}(s_{i'}) = \ell_{i'}^{(u),{j_0}}].
\end{align}

Let $\hat \cN$ be the network obtained by removing $A_{j_0}$ from $\cN$. We note that $\hat \cN$ is also a NDCS network. Let  $\hat P[ n_{j'}:\, j'\neq j]$ be the distribution on the outputs of $\hat \cN$ given by
$$\hat P[ n_{j'}:\, j'\neq j] = P_\token[ n_{j'}:\, j'\neq j_0 |\, n_{j_0}= n_{j_0}^{\min}].$$
We claim that $\hat P$ is again a TC distribution. Indeed, we claim that any of the two TC strategies for simulating $P_\token$ in the statement of the lemma, can be reduced to a TC strategy for simulating $\hat P$. 
To prove this, assume that a source $S_{i'}$ with $S_{i'}\to A_{j_0}$ only takes values $s_{i'}^{(u)}$ with $T_{i'}^{(u),j_0}(s_{i'}^{(u)})=\ell_{i'}^{(u),j_0}$. We assume that $S_{i'}$ takes such a value $s_{i'}^{(u)}$ with the conditional probability $\Pr[s_{i'}^{(u)}| T_{i'}^{(u),j_0}(s_{i'}^{(u)})=\ell_{i'}^{(u),j_0}]$. Sources not connected to $A_{j_0}$ and other parties behave as before. Then, using~\eqref{app:eq:prop-n-j-max} it is not hard to verify that the output distribution with this strategy equals $\hat P$. 

Therefore, we obtain two strategies for simulating the token counting distribution $\hat P$ on $\hat \cN$.  Now, since the number of parties in $\hat \cN$ is less than $I$, by the induction hypothesis the probability of distributing the tokens in the two strategies coincide. We note that $S_i$ was not connected to $A_{j_0}$ and its behavior does not change in the new strategies. Therefore, we have 
\begin{align*}
\Pr\big[T_i^{(1),j}(s_i^{(1)})=&t_i^j,\, \forall j:\, S_i\to A_j\big]\\ 
&= \Pr\big[T_i^{(2),j}(s_i^{(2)})=t_i^j,\, \forall j:\, S_i\to A_j\big],
\end{align*}
as desired.

\end{proof}

We now give the proof of Corollary~\ref{corollary:TCExtraMeasures} of the main text, which we rephrase for convenience:


\TCExtraMeasures*

\begin{proof}
We use the notation of Theorem~\ref{app:theorem:TC}.  We need to show that $\alpha_j = \{\eta^j_i:\, i\to j\}$ if and only if $T_i^j(s_i) = \eta_i^j$ for any $i\to j$. Suppose that  $A_j$ outputs $\alpha_j = \{\eta^j_i:\, i\to j\}$ and fix some source $S_{i}$ with $i\to j$. Suppose that $\eta_i^j> T_i^j(s_i)$ (the other case is similar).
Let $A_{j_1}, \dots, A_{j_k}$ be other parties connected to $S_{i}$. Suppose that the messages $s_{i'}$ of other sources $S_{i'}$ that are not connected to $A_j$, are chosen such that the sum of tokens received by $A_{j_1}, \dots, A_{j_k}$ from those sources is maximized. We note that by the NDCS assumption, these choices of $s_{i'}$'s do not affect the output of $A_j$. Let $m$ be this maximum number. Then, $m$ tokens are sent to $A_{j_1}, \dots, A_{j_k}$ by sources $S_{i'}\neq S_i$ and $\eta_i - T_i^j(s_i)$ tokens are sent by $S_i$. Therefore, we have
$$\sum_{\ell=1}^k n_{j_\ell} = m + \eta_i -T_i^j(s_i).$$
Now, $A_j$ claims that she has received $\eta_i^j$ tokens from $S_i$. This means that the sum of tokens received by $A_{j_1}, \dots, A_{j_k}$ and those received by $A_j$ from $S_i$ equals 
$$m + \eta_i -T_i^j(s_i) +\eta_i^j > m+\eta_i.$$
This is a contradiction since $m$ is the maximum number of possible tokens that can be ever sent to $A_{j_1}, \dots, A_{j_k}$ from sources $S_{i'}\neq S_i$, and $\eta_i$ is the number of tokens of $S_i$. This shows that $A_j$ outputs $\alpha_j = \{\eta^j_i:\, i\to j\}$ only if $T_i^j(s_i) = \eta_i^j$. For the other direction, that $\alpha_j = \{\eta^j_i:\, i\to j\}$ whenever $T_i^j(s_i) = \eta_i^j$, consider the probability of $\alpha_j = \{\eta^j_i:\, i\to j\}$.

\end{proof}

\section{Rigidity of CM distributions}\label{app:sec:CM}

In this section, we prove Theorem~\ref{theorem:CM} of the main text. 
Our proof relies on the Finner inequality and its equality condition \cite{Finner1992}. For self-containment, we reproduce a simplified version here with the network terminology, where we only specify the equality condition for indicator functions. 
\begin{theorem*}[Finner's inequality]
Let $\cN$ be a network admitting a PFIS by assigning $0<x_j<1$ to party $A_j$.
For any party $A_j$ let $g_j(\{s_i: i\rightarrow j\})$ be a real function of the messages she receives. Then, we have:
\begin{equation}\label{app:eq:finner-ineq-thm}
\E\Big[  \prod_j |g_j| \Big]\leq \prod_j \left(\E\big[  |g_j|^{\frac 1{x_j}}\big] \right)^{x_j}.
\end{equation}
In case of indicator functions $g_j(\{s_i: i\rightarrow j\})\in\{0,1\}$, equality holds in~\eqref{app:eq:finner-ineq-thm} if and only if there exist \emph{indicator} functions $\phi_i(s_i)\in\{0,1\}$ such that 
\begin{equation}\label{app:eq:finner-eq-thm}
g_j(\{s_i: i\rightarrow j\})=\prod_{i:i\rightarrow j} \phi_i (s_i), \quad \forall j.
\end{equation}
\end{theorem*}

In the following, we will use Finner's inequality for the indicator function $g_j^{(c)}$ corresponding to the color match $c$ being observed by party $A_j$. We will show that equality holds for these indicator function and the associated functions $\phi_i^{(c)}$ will indicate when source $S_i$ takes  color $c$.

Let us now rephrase Theorem~\ref{theorem:CM} of the main text for convenience:

\begin{theorem}[Color-Matching]\label{app:theorem:CM_Appendix}
Consider a network $\cN$ with parties $A_1, \dots,  A_{J}$ and sources $S_1, \dots,  S_{I}$. Assume that $\cN$ is an Exclusive Common-Source network that admits a Perfect Fractional Independent Set. 
Let $P=\{P(a_1, \dots,  a_J)\}$ be a Color-Matching distribution over $\cN$ in which any source $S_i$ takes color $c\in \{1,\dots, C\}$ with probability $p_{\match}(c)>0$. 

Now consider another strategy that simulates $P$ on $\cN$. Then, this strategy is a Color-Matching strategy with the same color distribution as before. More precisely, let $g_j^{(c)}\in\{0,1\}$ be the indicator function that $A_j$ outputs color match $c$. Then, for any color $c\in\{1,\dots, C\}$ there is an indicator function $\phi_i^{(c)}\in\{0,1\}$ such that
\begin{enumerate}
\item[{\rm (i)}] $g_j^{(c)}=\prod_{i:i\rightarrow j} \phi_i^{(c)}$
\item[{\rm (ii)}]$\forall s_i, \sum_c \phi_i^{(c)}(s_i)=1$
\item[{\rm (iii)}] $\E\big[\phi_i^{(c)}\big] = p_\match(c)$
\end{enumerate}
\end{theorem}

In this theorem, $\phi_i^{(c)}(s_i)=1$ when source $S_i$ takes color $c$. (ii) says that any possible message $s_i$ of source $S_i$ is associated to a unique color $c$. (i) indicates that $A_j$ outputs color match $c$ if and only if all the sources connected to her take color $c$.  Finally, (iii) implies that the sources take colors with the same probability distribution as in the original strategy.

\begin{proof}
To prove (i) we use Finner's inequality. Let $\{x_j: 1\leq j\leq J\}$ be a PFIS of $\cN$. Then, the probability that all parties output color match $c$ is equal to the probability that all sources take color $c$, i.e.,
\begin{align*}
\Pr(c,\cdots,c) &= \prod_{i=1}^I p_\match(c) 
= \prod_{i=1}^I \prod_{j: i\to j} p_\match(c)^{x_j} \\
& = \prod_{j=1}^J \prod_{i: i\to j} p_\match(c)^{x_j} 
= \prod_j \Pr(A_j=c)^{x_j}.
\end{align*} 
On the other hand, we have $\E[\prod_j g_j^{(c)}]=\Pr(c,\cdots,c)$ and 
 $$\E\Big[ |g_j^{(c)}|^{\frac{1}{x_j}}\Big] = \E[g_j^{(c)}] = \Pr(A_j=c).$$ 
 Therefore, the Finner inequality~\eqref{app:eq:finner-ineq-thm}, turns into an equality and by the equality condition~\eqref{app:eq:finner-eq-thm} functions $\phi_i^{(c)}$ satisfying (i) exist. 
 
To prove (ii) we first show that $\sum_c \phi_i^{(c)}\leq 1$, which since $\phi_i^{(c)}$ takes values in $\{0,1\}$, means that for any $s_i$ there is \emph{at most} one color $c$ for which $\phi_i^{(c)}(s_i)=1$. To this end, assume that there are $c_0\neq c_1$ and $s_i^*$ such that $\phi_i^{(c_0)}(s_i^*)=\phi_i^{(c_1)}(s_i^*)=1$. Let $A_{j_0}, A_{j_1}$ be the two parties whose unique common source is $S_i$. For any source $S_{i_0}\neq S_i$ with $i_0\to j_0$ let $s_{i_0}^*$ be such that $\phi_{i_0}^{(c_0)}(s_{i_0^*}) = 1$. We note that such $s_{i_0}^*$ exists since $\Pr(A_{j_0}=c_0)>0$. Similarly, choose $s_{i_1}^*$ for any source $S_{i_1}\neq S_i$ with $i_1\to j_1$ such that $\phi_{i_1}^{(c_1)}(s_{i_1}^*) = 1$. Then, with these choices of $s_i^*$'s, using (i) we find that $\Pr(A_{j_0}=c_0, A_{j_1}=c_1)>0$. However, in a CM distribution if two parties share a source, they can never output a color match with different colors. Therefore, $\sum_c \phi_i^{(c)}\leq 1$.  

For any source $S_i$ let 
$$q_i(c) := \E[\phi_i^{(c)}] = \Pr[\phi_i^{(c)} =1].$$
Then, by $\sum_c \phi_i^{(c)}\leq 1$ we have
\begin{align}\label{app:eq:q-i-c-prob}
\sum_c q_i(c)=\E\Big[\sum_c \phi_i^{(c)}\Big] \leq 1.
\end{align}
If we show that equality holds in the above equation, (ii) is proven. To this end, note that by (i) we have
$$p_\match(c)^I=\Pr(c, \dots,  c) = \prod_i \Pr[\phi_i^{(c)}=1] = \prod_i q_i(c).$$
Then, equality in~\eqref{app:eq:q-i-c-prob} as well as (iii) are derived from Lemma~\ref{app:lem:color-dist-match} below.
\end{proof}

\begin{lemma}\label{app:lem:color-dist-match}
Let $p_\match(c)>0$ be a probability distribution over $\{1, \dots,  C\}$. Also, let $q_i(c)\geq 0$, for $i\in \{1, \dots,  I\}$, be such that for any $i$,  $\sum_c q_i(c)\leq 1$. Moreover, assume that for any $c$ we have
\begin{align}\label{app:eq:lem-qic}
\prod_{i=1}^I q_i(c) = p_\match(c)^I.
\end{align}
Then, $q_i(c)=p_\match(c)$ for any $i$ and $c$. In particular, we have $\sum_c q_i(c)=1$ for all $i$.
\end{lemma}

\begin{proof}
Define $f_i:\{1, \dots,  C\}\to \mathbb R$ by 
$$f_i(c) = \Big(\frac{q_i(c)}{p_\match(c)}\Big)^{1/I}.$$ 
We compute
\begin{align*}
1&= \sum_{c} p_\match(c)=  \sum_{c} \Big[\prod_{i=1}^I q_i(c) \Big]^{1/I}
 = \sum_{c}\, p_\match(c) \prod_{i=1}^I f_i(c) \\
& = \E\Big[  \prod_{i=1}^I f_i\Big],
\end{align*}
where the expectation is with respect to the distribution $p_\match(c)$. Then, by H\"older's inequality we have
\begin{align*}
 1 &\leq \prod_{i=1}^I \|f_i\|_{I} = \prod_{i=1}^I \E\big[f_i^I\big]^{1/I}
 = \prod_{i=1}^I \E\big[q_i/p\big]^{1/I}\\
& =\prod_{i=1}^I \Big(\sum_c q_i(c)\Big)^{1/I}\leq 1.
\end{align*} 
Therefore, H\"older's inequality and inequalities $\sum_c q_i(c)\leq 1$ are equalities. Therefore, all the functions $f_i$, and then $q_i/p$'s and $q_i$'s  are collinear.  Then, using the normalization $\sum_c q_i(c)=1$, which we just proved, we find that $q_i$'s are equal. Using this in~\eqref{app:eq:lem-qic} we obtain $q_i(c)=p_\match(c)$ as desired. 
\end{proof}

We now give the proof of an extension of Corollary~\ref{corollary:CMExtraMeasures} of the main text:

\begin{corollary}[Refined measurements in CM]\label{app:corollary:CMExtraMeasures}
Let $\cN$ be an ECS network admitting a PFIS and let $P$ be the outcome distribution of a \emph{quantum} CM strategy. Suppose that $\ket{c_i: j\to i}$, for some $1\leq c_i\leq C$ belongs to the measurement basis of a party $A_j$. Then, the followings hold:
\begin{itemize}
\item[{\rm{(i)}}] Assume that for any source $S_i$ connected to $A_j$ there is a party $A_{j^{(i)}}$ with $S_i\to A_{j^{(i)}}$ such that $S_i$ is the unique common source of $A_j$ and $A_{j^{(i)}}$.
Then, $\ket{c_i:\, i\to j}$ is rigid for $A_j$.

\item[{\rm{(ii)}}] Let $A_{j_1}, \dots, A_{j_k}$ be a list of parties (different from $A_j$) such that for any source $S_i$ with $i\to j$ there is $\ell$ with $i\to j_\ell$. Let $\mathcal S$ be the union of the set of sources connected to $A_{j_1}, \dots, A_{j_k}$ which by the previous assumption includes $S_i$'s with $i\to j$. Let $\{c_{i'}:\, i'\in \mathcal S\}$ be an extension of $\{c_i:\, i\to j\}$ that assigns colors to all sources of $\mathcal S$. Suppose that for any $\ell$, the computational basis state $\ket{c_{i'}: i'\to j_\ell}$ is rigid for $A_{j_\ell}$. Then, $\ket{c_i:\, i\to j}$ is rigid for $A_j$.

\item[{\rm{(iii)}}] Suppose that there is an extension $\{c_{i'}: 1\leq i'\leq I \}$ of $\{c_i:\, i\to j\}$ that assigns a color to any source in $\cN$, such that for any party $A_{j'}$ the computational basis state $\ket{c_{i'}:\, i'\to j'}$ belongs to the measurement basis of $A_{j'}$. Then, $\ket{c_i:\, i\to j}$ is rigid for $A_j$.
\end{itemize}
\end{corollary}

\begin{proof}
To prove this corollary we use Theorem~\ref{app:theorem:CM_Appendix} and the notation developed there. 

\medskip
\noindent (i) We need to show that $A_j$ outputs $\{c_i: j\to i\}$ if and only if $\phi_i^{(c_i)}(s_i)=1$ for any $i\to j$. Suppose that for such an $i$ we have $\phi_i^{(c_i)}(s_i)=0$ and $\phi_i^{(c'_i)}(s_i)=1$ for some $c'_i\neq c_i$. As in the statement of the corollary, party $A_{j^{(i)}}$ has the property that $S_i$ is the unique common source of $A_j$ and $A_{j^{(i)}}$. Then, we may choose $s_{i'}$'s for any $S_{i'}\neq S_i$ with $S_{i'}\to A_{j^{(i)}}$ such that $\phi_{i'}^{(c'_i)}(s_{i'}) =1$. We note that the choice of such $s_{i'}$'s does not affect the output of $A_j$.  Thus, we note that any source connected to $A_{j^{(i)}}$ takes color $c'_i$. This means that $A_{j^{(i)}}$ outputs color match $c'_i$. On the other hand, by assumption $A_j$ claims that $S_i$ takes color $c_i\neq c'_i$. This is a contradiction. As a result, $A_j$ outputs $\{c_i: j\to i\}$  only if $\phi_i^{(c_i)}(s_i)=1$ for any $i\to j$. For the other direction that $A_j$ outputs $\{c_i: j\to i\}$  if $\phi_i^{(c_i)}(s_i)=1$ for any $i\to j$, consider the probability of outputting  $\{c_i: j\to i\}$.

\medskip
\noindent (ii) Let $\{s_i:\, i\to j\}$ be a list of messages taken by sources connected to $A_j$ such that $\phi_i^{(c_i)}(s_i)=1$. We need to show that in this case $A_j$ outputs $\{c_i:\, i\to j\}$. For any other  $i'\in \mathcal S$ choose $s_{i'}$ such that $\phi_{i'}^{(c_{i'})}(s_{i'})=1$. Then, by the rigidity assumptions $A_{j_\ell}$, for any $1\leq \ell\leq k$, outputs $\{c_{i'}:\, i'\to j_\ell\}$. Therefore, since any source $S_i$ connected to $A_j$ is connected to at least one of $A_{j_\ell}$'s, the color taken by $S_i$ is determined by the outputs of $A_{j_1}, \dots, A_{j_\ell}$. Thus, since $\ket{c_i:\, i\to j}$ belongs to her measurement basis,  $A_j$ has no choice but outputting this list of colors. For the other direction that $A_j$ outputs $\{c_i:\, i\to j\}$ only if the connected sources take these colors consider the probability of output $\{c_i:\, i\to j\}$.

\medskip
\noindent (iii) Let $f_{j'}$ be the indicator function that $A_{j'}$ outputs $\{c_{i}:\, i\to j'\}$. Then, using Finner's inequality for $f_{j'}$'s as in the proof of Theorem~\ref{app:theorem:CM_Appendix}, we find that there are 0/1-valued functions $\psi_{i}$ such that $A_{j'}$ outputs  $\{c_{i}:\, i\to j'\}$ if and only if $\psi_{i}(s_{i})=1$. We need to show that $\psi_{i}(s_{i})=1$ if and only if $\phi_{i}^{(c_{i})}(s_i)=1$, which means that $A_{j'}$ outputs $\{c_{i}:\, i\to j'\}$ if and only if $S_{i}$ with $i\to j'$, takes color $c_{i}$. 

Fix a source $S_{i}$ and assume that $\phi_{i}^{(c'_{i})}(s_i)=1$ for some $ c'_{i}\neq c_{i}$. Using  the ECS assumption let $A_{j_1}, A_{j_2}$ be two parties connected to $S_{i}$ such that $S_{i}$ is their unique common source. Fix the message of sources connected to $A_{j_1}$ (including $S_{i}$) as before so that $A_{j_1}$ outputs $\{c_{i'}:\, i'\to j\}$. Next, choose the messages of sources $S_{i'}\neq S_i$ connected to $A_{j_2}$ such that she outputs color match $c'_i$.  We note that such a choice is feasible since $S_i$ is the unique common source of $A_{j_1}, A_{j_2}$ and  $\phi_{i}^{(c'_{i})}(s_i)=1$. This is a contradiction since now $A_{j_1}$ claims that $S_i$ takes color $c_i$, but $A_{j_2}$ claims that it takes color $c'_i$.

\end{proof}

\section{All ring scenarios with bipartite sources}\label{app:sec:ring}

We start by the proof of Claim~\ref{claim:RingScenario} of the main text.
\begin{proof}[Proof of Claim~\ref{claim:RingScenario} of the main text]
{\rm{(i)}} We compute:
\begin{align}
q(r_1, \dots,  r_n) &= \Pr\big(\alpha_j=\ket{v_{j, r_j}},\, \forall j\,\big| n_j =1,\, \forall j \big  )\nonumber\\
& = \frac{1}{\Pr[ n_j =1,\, \forall j]} \Pr\big(\alpha_j=\ket{v_{j, r_j}},\, \forall j  \big  )\nonumber\\
& = 2^{n-1} \Pr\big(\alpha_j=\ket{v_{j, r_j}},\, \forall j  \big  )\nonumber\\
& = \frac{1}{2} \Big|\prod_{j} \omega^{(1)}_{j, r_j} + \prod_{j} \omega^{(2)}_{j, r_j}  \Big|^2\label{app:eq:margin-00n}
\end{align}

\noindent
{\rm{(ii)}} We concentrate on the case $t=1$, the other case being similar. For any $j$, we have
\begin{align*}
q(r_j, t=1) & = \Pr\big(\alpha_j=\ket{v_{j, r_j}}, t=1\,\big| n_i =1,\, \forall i \big  )\\
& = 2^{n-1}\Pr\big(\alpha_j=\ket{v_{j, r_j}}, t=1 \big  )\\
& = 2^{n-1}\Pr\big(\alpha_j=\ket{v_{j, r_j}}, S_i \leadsto A_{i},\, \forall i \big  ),
\end{align*}
where by $S_i \leadsto A_{i}$ we mean that $S_i$ sends his token to $A_{i}$. We continue
\begin{align*}
&q(r_j, t=1)   =2^{n-1}\Pr\big(\alpha_j=\ket{v_{j, r_j}}, S_i \leadsto A_{i},\, \forall i \big  )\\
& =2\Pr\big(\alpha_j=\ket{v_{j, r_j}}, S_i \leadsto A_{i},\,  i=j, j-1 \big  )\\
& =4\Pr\big(\alpha_j=\ket{v_{j, r_j}}, S_i \leadsto A_{i},\,  i=j, j-1, S_{j+1}\leadsto A_{j+2} \big  ),
\end{align*}
where we use the fact that $A_j$'s output is independent of whether $S_i$ for $i\notin\{j, j-1\}$ sends the token to $A_i$ or $A_{i+1}$. 
Now, assume that $S_i \leadsto A_{i}$ for  $i=j, j-1$ and $S_{j+1}\leadsto A_{j+2}$. 
In this case, $A_{j+1}$ receives no token, hence $a_{j+1} = \ket{00}$ (i.e., $n_{j+1}=0$). 
Conversely, when $a_{j+1} = \ket{00}$ and $\alpha_j= \ket{v_{j, r_j}}$, the distribution of tokens by sources $S_{j-1}, S_j$ and $S_{j+1}$ is  $S_i \leadsto A_{i}$ for  $i=j, j-1$ and $S_{j+1}\leadsto A_{j+2}$. Therefore, we have
\begin{align}
q(r_j, t=1) &= 4\Pr\big(\alpha_j=\ket{v_{j, r_j}}, a_{j+1} = \ket{00} \big  )\nonumber\\
& = \frac{1}{2} |\omega_{j, r_j}^{(1)}|^2.\label{app:eq:margin-00t}
\end{align}
\end{proof}

The rest of this section is devoted to the proof of Proposition~\ref{propo:RingScenario} of the main text. 
We show that for certain \emph{asymptotic} choice of parameters $\omega^{(t)}_{j, r_j}$, which we assume to be real, the LP given by Claim~\ref{claim:RingScenario} of the main text is infeasible. This proposition is indeed proven in~\cite{Renou2019a} when $n$ is odd. Here, for the sake of completeness we include this case as well.

\medskip
\noindent
\emph{Proof of Proposition~\ref{propo:RingScenario} of the main text.}
Assume that the parameters $\omega^{(t)}_{j, r_j}$ are all real.
Let us define $x_{r_1, \dots,  r_n}$ by
\begin{align*}
q(r_1, \dots,  r_n,& t=1) =\\
 &\frac{1}{2} \Big(\prod_j \big(\omega_{j, r_j}^{(1)}\big)^2 + \prod_j \omega_{j, r_j}^{(1)}\omega_{j, r_j}^{(2)} + x_{r_1, \dots,  r_n}\Big)
\end{align*} 
Then, by~\eqref{app:eq:margin-00n} we have
\begin{align*}
q(r_1, \dots,  r_n,& t=2) =\\
 &\frac{1}{2} \Big(\prod_j \big(\omega_{j, r_j}^{(2)}\big)^2 + \prod_j \omega_{j, r_j}^{(1)}\omega_{j, r_j}^{(2)} - x_{r_1, \dots,  r_n}\Big)
\end{align*} 
Moreover, by~\eqref{app:eq:margin-00t} and the fact that $\ket{v_{i, 1}}$ and $\ket{v_{i, 2}}$ are orthonormal, we have
\begin{align}\label{app:eq:sum-x-0}
\sum_{r_i:\, i\neq j} x_{r_1, \dots,  r_n} =0, \qquad \forall r_j.
\end{align}
Observe that the non-negativity of $q(r_1, \dots,  r_n, t)$ gives
\begin{align*}
\prod_j \big(\omega_{j, r_j}^{(2)}\big)^2 + \prod_j &\omega_{j, r_j}^{(1)}\omega_{j, r_j}^{(2)} \geq x_{r_1, \dots,  r_n}\\
&\geq-\prod_j \big(\omega_{j, r_j}^{(1)}\big)^2 - \prod_j \omega_{j, r_j}^{(1)}\omega_{j, r_j}^{(2)}.
\end{align*}
To simplify the notation, let $\omega_{j, 1}^{(1)} = -\omega_{j, 2}^{(2)}=\lambda_j $ and $\omega_{j, 1}^{(0)} = \omega_{j, 0}^{(1)}=\mu_j$ with $\lambda_j, \mu_j>0$ and $\lambda_j^2+\mu_j^2=1$. Then, the above inequality turns to 
\begin{align}
\prod_j \lambda_{j}^{2r_j}&\mu_j^{2(1-r_j)} +(-1)^{S} \prod_j \lambda_{j}\mu_{j} \geq x_{r_1, \dots,  r_n}\nonumber\\
&\geq-\Big(\prod_j \lambda_{j}^{2(1-r_j)}\mu_j^{2r_j} +(-1)^S \prod_j \lambda_{j}\mu_{j}\Big),\label{app:eq:ineq-S-x-alpha}
\end{align}
where $S=S(r_1, \dots,  r_n)= \sum_j r_j$. 

In the following, we show that for appropriate choices of $\lambda_j, \mu_j$ equations~\eqref{app:eq:sum-x-0} and~\eqref{app:eq:ineq-S-x-alpha} do not have a solution which by the above discussion means that there is no classical strategy simulating the quantum distribution. 

\medskip
Let us assume that $\lambda_j=\lambda$ and $\mu_j=\mu$ for all $j$. Then,~\eqref{app:eq:ineq-S-x-alpha} becomes
\begin{align}
\lambda^{2S}\mu^{2(n-S)} +&(-1)^{S} \lambda^n\mu^n \geq x_{r_1, \dots,  r_n}\nonumber\\
&\geq-\Big(\lambda^{2(n-S)}\mu^{2S} +(-1)^S \lambda^n\mu^n\Big).\label{app:eq:alpha-beta-n-ineq-00}
\end{align}
Observe that in the above inequality the upper and lower bounds on 
$x_{r_1, \dots,  r_n}$ depend only on $S = \sum_j r_j$. Thus, let us define
$$x_{S} = \frac{1}{\binom{n}{S}} \sum_{r_1+\cdots +r_n=S} x_{r_1, \dots,  r_n}. $$
Then, $x_S$ satisfies 
\begin{align}
\lambda^{2S}\mu^{2(n-S)} +&(-1)^{S} \lambda^n\mu^n \geq x_{S}\nonumber\\
&\geq-\Big(\lambda^{2(n-S)}\mu^{2S} +(-1)^S \lambda^n\mu^n\Big).\label{app:eq:ineq-S-x-alpha-unif}
\end{align}
Moreover, summing~\eqref{app:eq:sum-x-0} over all $j$ with fixed $r_j=r\in \{0,1\}$, we obtain
\begin{align}\label{app:eq:sum-x-0-unif}
\sum_{S=0}^{n-1}  \binom{n-1}{S}  x_{S+r} =0, \qquad r=0, 1.
\end{align}
Thus, we need to show that~\eqref{app:eq:ineq-S-x-alpha-unif} and~\eqref{app:eq:sum-x-0-unif}  do not have a solution. 

Let us assume that $\lambda=\epsilon$ is small and $\mu=\sqrt{1-\lambda^2}$. Then, taking the leading term in $\epsilon$ (i.e., $\epsilon^n$) and replacing $x_S$ by $x_S=\epsilon^{-n}x_S$ (notice that~\eqref{app:eq:sum-x-0-unif} is still satisfied) we find that 
\begin{align}
(-1)^{S}\geq x_S, & \qquad S>n/2,\label{app:eq:S-geq-n2}\\
 x_S\geq - (-1)^S, & \qquad S<n/2.\label{app:eq:S-leq-n2}
\end{align}
In the following, by separating even and odd cases we show that for any $n\geq 3$ the above inequalities are infeasible.

\subsubsection{$n$ odd}
As mentioned before, for odd $n$ the fact that~\eqref{app:eq:sum-x-0-unif},~\eqref{app:eq:S-geq-n2} and~\eqref{app:eq:S-leq-n2} do not have a solution is already proven in~\cite{Renou2019a}. Here, for the sake of completeness we reproduce a proof.

Let $C>0$ be such that 
\begin{align}
&C\binom{n-1}{S} -\binom{n-1}{S-1}\geq 0, \qquad \forall S<n/2,\label{app:eq:const-C-1}\\
&C\binom{n-1}{S} -\binom{n-1}{S-1}\leq 0, \qquad \forall S>n/2\label{app:eq:const-C-2}.
\end{align} 
Then, multiply equation~\eqref{app:eq:sum-x-0-unif} for $r=0$ by $C$ and subtract it from the same equation for $r=1$. We obtain
\begin{align}
0& = C\sum_{S=0}^{n-1} \binom{n-1}{S} x_S - \sum_{S=1}^n \binom{n-1}{S-1}x_S\nonumber\\
& = \sum_{S=0}^{n} \left( C\binom{n-1}{S} - \binom{n-1}{S-1}     \right)x_S.\label{app:eq:mult-C-subtract}
\end{align}
Therefore, using~\eqref{app:eq:S-geq-n2} and~\eqref{app:eq:S-leq-n2} and the constrains we put on $C$ we have
\begin{align*}
0&\geq -\sum_{S=0}^{(n-1)/2} \left( C\binom{n-1}{S} - \binom{n-1}{S-1}     \right)(-1)^S\\
& \qquad + \sum_{S=(n+1)/2}^{n} \left( C\binom{n-1}{S} - \binom{n-1}{S-1}     \right)(-1)^S\\
& =  -\sum_{S=0}^{(n-1)/2} \left( C\binom{n-1}{S} - \binom{n-1}{S-1}     \right)(-1)^S\\
& \qquad + \sum_{S=0}^{(n-1)/2} \left( C\binom{n-1}{S-1} - \binom{n-1}{S}     \right)(-1)^{n-S}\\
& \stackrel{(a)}{=}  \sum_{S=0}^{(n-1)/2} \left( C\binom{n}{S} - \binom{n}{S}     \right)(-1)^{n-S}\\
& \stackrel{(b)}{=} (C-1)(-1)^{(n+1)/2}\binom{n-1}{(n-1)/2},
\end{align*}
where in (a) we use the fact that $n$ is odd and $-(-1)^S=(-1)^{n-S}$ as well as Pascal's rule. Moreover, for (b) we use Pascal's rule to obtain
\begin{align}\label{eq:pascal-sum}
\sum_{S=0}^K (-1)^S\binom{m}{S} & = \sum_{S=0}^K (-1)^S\left(\binom{m-1}{S} +\binom{m-1}{S-1}   \right)\nonumber\\
& = (-1)^R\binom{m-1}{R}.
\end{align} 
Now, note that for any $n$ and sufficiently small $\delta_n>0$, both values of $C=1+\delta_n$ and $C=1-\delta_n$ satisfy~\eqref{app:eq:const-C-1} and~\eqref{app:eq:const-C-2}. Thus, we obtain
$$0\geq \pm \delta_n(-1)^{(n+1)/2}\binom{n-1}{(n-1)/2},$$
that is a contradiction.

\subsubsection{$n$ even}

For even $n$ it is more convenient to separate the cases of $n=4k$ and $n=4k+2$.

Let us first assume that $n=4k$ for $k\geq 1$. Taking the difference of~\eqref{app:eq:sum-x-0-unif} for $r=0$ and $r=1$ and using~\eqref{app:eq:S-geq-n2} and~\eqref{app:eq:S-leq-n2} we obtain
\begin{align*}
0 & = \sum_{S=0}^n \left( \binom{n-1}{S-1} - \binom{n-1}{S}   \right) x_S\\
& = \sum_{S=0}^{2k-1} \left( \binom{n-1}{S-1} - \binom{n-1}{S}   \right) x_S\\
&\qquad +  \sum_{S=2k+1}^{4k} \left( \binom{n-1}{S-1} - \binom{n-1}{S}   \right) x_S \\
& \leq -\sum_{S=0}^{2k-1} \left( \binom{n-1}{S-1} - \binom{n-1}{S}   \right) (-1)^S\\
&\qquad +  \sum_{S=2k+1}^{4k} \left( \binom{n-1}{S-1} - \binom{n-1}{S}   \right) (-1)^S,
\end{align*}
where for the second equality we use $\binom{n-1}{2k} = \binom{n-1}{2k-1}$, and we carefully checked the sign of the expressions to derive the inequality. 
Then, we have
\begin{align*}
0 &\leq -\sum_{S=0}^{2k-1} \left( \binom{n-1}{S-1} - \binom{n-1}{S}   \right) (-1)^S\\
&\qquad +  \sum_{S=0}^{2k-1} \left( \binom{n-1}{S} - \binom{n-1}{S-1}   \right) (-1)^S \\
&=2\sum_{S=0}^{2k-1} \left(   \binom{n-1}{S}-\binom{n-1}{S-1}   \right) (-1)^S\\
& = 2 \sum_{S=0}^{2k-2} \binom{n-1}{S} \big( (-1)^S - (-1)^{S+1} \big) \\
&\qquad+ 2\binom{n-1}{2k-1}(-1)^{2k-1}\\
& = 4 \sum_{S=0}^{2k-2} (-1)^S\binom{n-1}{S}   - 2\binom{n-1}{2k-1}\\
& \stackrel{(a)}{=} 4\binom{n-2}{2k-2} -2\binom{n-1}{2k-1}\\
& \stackrel{(b)}{=} -2\Big(\frac{4k-1}{2k-1} -2\Big)\binom{n-2}{2k-2},
\end{align*}
that is a contradiction. Here, for (a) we use~\eqref{eq:pascal-sum}, and  for (b) we use $\binom{n-1}{2k-1} = \frac{n-1}{2k-1}\binom{n-2}{2k-2}$.

\medskip
It is not hard to verify that the above argument does not work for $n=4k+2$. Indeed, it can be shown that equations~\eqref{app:eq:sum-x-0-unif},~\eqref{app:eq:S-geq-n2} and~\eqref{app:eq:S-leq-n2} are feasible when $n=6$. Even replacing~\eqref{app:eq:sum-x-0-unif} with~\eqref{app:eq:ineq-S-x-alpha-unif} which is its origin, equations~\eqref{app:eq:sum-x-0-unif} and~\eqref{app:eq:ineq-S-x-alpha-unif} are again feasible for $n=6$. Thus, we need to take a different path for $n=4k+2$. 

One approach is to derive stronger equations than~\eqref{app:eq:margin-00t} for marginals of $q(r_1, \dots,  r_n, t)$. Indeed, it can be shown that for instance 
$$q(r_1, \dots,  r_{n-1}, t) = \frac{1}{2}\Big|\prod_{j=1}^{n-2} \omega_{j, r_j}^{(t)} \Big|^2.$$
Such equations give stronger constraints on $x_S$'s comparing to~\eqref{app:eq:sum-x-0-unif}.  This approach does give a proof of the result for $n=4k+2$ when $k>1$. For $n=6$, however, the resulting LP is again feasible for all choices of $\lambda, \mu$.  Therefore, we do not give the details of this approach, and instead take a different path.

Instead of assuming that all $\lambda_j$'s are equal, we assume that $\lambda_1=\cdots=\lambda_{n-1} = \epsilon$, $\mu_1=\cdots= \mu_{n-1} = \sqrt{1-\epsilon^2}$ and $\lambda_n=\mu_n=1/\sqrt 2$. Then, scaling $x_{r_1, \dots,  r_n}$ with $2\epsilon^{-(n-1)}$ for sufficiently small $\epsilon>0$ the inequality~\eqref{app:eq:ineq-S-x-alpha} gives
\begin{align*}
&(-1)^{S+r_n} \geq x_{r_1,\dots,  r_n}, \qquad  S>\frac{n-1}{2},\\
&-(-1)^{S+r_n} \leq x_{r_1,\dots,  r_n}, \qquad  S<\frac{n-1}{2},
\end{align*}
where $S= r_1+\cdots+r_{n-1}$. Using the idea of symmetrization as before, we find that there are $x_{S, r_n}$ such that 
\begin{align}
&(-1)^{S+r_n} \geq x_{S, r_n}, \qquad  S>\frac{n-1}{2},\label{app:eq:s-r-n-n-1}\\
&-(-1)^{S+r_n} \leq x_{S, r_n}, \qquad  S<\frac{n-1}{2},\label{app:eq:s-r-n-n-2}
\end{align}
Moreover,~\eqref{app:eq:sum-x-0} gives 
\begin{align}\label{app:eq:x-sum-s-rn}
\sum_{S=0}^{n-1}\binom{n-1}{S} x_{S, r_n} =0, \qquad r_n=0, 1,
\end{align}
and
\begin{align*}
\sum_{S=0}^{n-2} \binom{n-2}{S}(x_{S+r, 0} +x_{S+r, 1})= 0, \qquad r=0,1.
\end{align*}
Subtracting the above equations for $r=0$ and $r=1$, we obtain
\begin{align*}
0 & = \sum_{S=0}^{n-2} \binom{n-2}{S}(x_{S, 0} +x_{S, 1})\\
&\quad - \sum_{S=0}^{n-2} \binom{n-2}{S}(x_{S+1, 0} +x_{S+1, 1})\\
& = \sum_{S=0}^{n-1} \Big[ \binom{n-2}{S} - \binom{n-2}{S-1}   \Big]  (x_{S, 0} +x_{S, 1}).
\end{align*}
Next, using~\eqref{app:eq:s-r-n-n-1} and~\eqref{app:eq:s-r-n-n-2}  we find that $(x_{S, 0} +x_{S, 1})\leq 0$ if $S>(n-1)/2$, and  $(x_{S, 0} +x_{S, 1})\geq 0$ if $S<(n-1)/2$. As a result, all terms in the above sum are non-negative. Therefore, since their sum is zero, all the inequalities in~\eqref{app:eq:s-r-n-n-1} and~\eqref{app:eq:s-r-n-n-2} are equalities. However, for these choices of $x_{S, r_n}$, using~\eqref{app:eq:x-sum-s-rn}, we have
\begin{align*}
0& = \sum_{S=0}^{n-1}\binom{n-1}{S} x_{S, 0} \\
& =  \sum_{S=0}^{n/2-1}\binom{n-1}{S} (-1)^{S+1} +  \sum_{S=n/2+1}^{n-1}\binom{n-1}{S} (-1)^S\\
& = \sum_{S=0}^{n/2-1}\binom{n-1}{S} (-1)^{S+1} + \sum_{S=0}^{n/2-2}\binom{n-1}{S} (-1)^{S+1} \\
& = 2 \sum_{S=0}^{n/2-2}\binom{n-1}{S} (-1)^{S+1}  + \binom{n-1}{n/2-1}(-1)^{n/2}\\
& = 2 \sum_{S=0}^{n/2-2}\binom{n-1}{S} (-1)^{S+1}  + \binom{n-1}{n/2-1}(-1)^{n/2}\\
& = \Big(2 - \frac{n-1}{n/2-1}\Big)\binom{n-2}{n/2-2},
\end{align*}
that is a contradiction. Here, in the last equation we use~\eqref{eq:pascal-sum} and $ \binom{n-1}{n/2-1} =  \frac{n-1}{n/2-1}\binom{n-2}{n/2-2}$.

\section{Bipartite-sources complete networks}\label{app:sec:BipartiteSourcesCompleteNet}

We prove here that the distribution $q(r_1, \dots,  r_n, t)$ given by
\begin{align*}
\Pr\Big(a_j = \ket{v_{j, r_j}}\,\forall j,\, t \Big|   \text{ ambiguous outputs} \Big).
\end{align*}
satisfies Claim~\ref{claim:RingScenario} of the main text.
The first equality in the claim is derived using 
$$q(r_1, \dots,  r_n)  = \Pr\big(a_j=\ket{v_{j, r_j}}, \forall j\big),$$
and computing the quantum probabilities.
For the second equality we compute:
\begin{align*}
q(r_j, t) & = \Pr\big(a_j= \ket{v_{j, r_j}}, t \big| \text{ ambiguous outputs}  \big)\\
& = \frac{1}{\Pr(\text{ambiguous outputs})}\Pr\big(a_j= \ket{v_{j, r_j}}, t \big)\\
& = 2^{\binom{n}{2}-1}\Pr\big(a_j= \ket{v_{j, r_j}}, t \big)\\
& = 2^{n-2}\Pr\big(a_j= \ket{v_{j, r_j}}, c_{jj'} = t \text{ iff } j\in\{j\pm 1\} \big)\\
& = 2^{n-2}\Pr\big(a_j= \ket{v_{j, r_j}}, c_{j, (j+1)} = t \big)\\
& = 2^{2(n-2)}\Pr\big(a_j= \ket{v_{j, r_j}}, c_{j', (j+1)} = t, \forall j' \big)\\
& = 2^{2(n-2)}\Pr\big(a_j= \ket{v_{j, r_j}}, a_{j+1} = \ket{t\cdots t} \big)\\
& = \frac{1}{2} \big|\omega_{j, r_j}^{(t)}\big|^2,
\end{align*}
where $c_{jj'}$ is the color taken by source $S_{jj'}$.

\section{CM scenario via graph coloring}\label{app:sec:GraphColoring}

Recall that our network comes from the complete graph where the $n$ vertices are associated with sources $S_1, \dots, S_n$ and edges are associated with parties $A_{ij}$. This network satisfies ECS and admits a PFIS.
Consider the following quantum CM strategy on this network:
\begin{itemize}
\item All the sources distribute $\frac{1}{\sqrt n}\sum_{c=1}^n\ket{c}^{\otimes n}$.
\item The measurement basis of $A_{ij}$ consists of vectors 
\begin{align*}
&\ket{c, c}:\,~ 1\leq c\leq n\\
& \ket{c_i, c_j}:\,~ c_i+c_j\notin \{i+j, 2(n+1)-(i+j)\}\\
& \ket{v_{ij, r}}=\sum_{c_i+c_j\in \cS_{ij}} \omega_{ij, r}^{(c_i,c_j)} \ket{c_i, c_j}:\, 1\leq r\leq R_{ij},
\end{align*}
where $R_{ij}$ is the number of pairs $(c_i, c_j)$ satisfying $c_i\neq c_j$ and $c_i+c_j\in \{i+j, 2(n+1)-(i+j)\}$.
\end{itemize}

We claim that for $n\geq 5$ and an appropriate choice of parameters $ \omega_{ij, r}^{(c_i,c_j)}$ the resulting CM distribution is nonlocal.  Suppose that a classical strategy, which by Theorem~\ref{theorem:CM} of the main text is necessarily a CM strategy, simulates this distribution. Using part (i) of Corollary~\ref{app:corollary:CMExtraMeasures} (an extension of Corollary~\ref{corollary:CMExtraMeasures} of the main text) we find that all the basis vectors $\ket{c_i, c_j}$with $c_i+c_j\notin \{i+j, 2(n+1)-(i+j)\}$ are rigid for $A_{ij}$. That is, $A_{ij}$ outputs $\ket{c_i, c_j}$ for such pairs if and only if she receives colors $c_i$ and $c_j$ from $S_i$ and $S_j$, respectively. We then restrict to the case where all the parties' outputs are ambiguous. That is, we assume that  $A_{ij}$, for any $i<j$, outputs $\ket{v_{ij, r_{ij}}}$ for some $r_{ij}$. 

\begin{claim}\label{claim:CMKn}
Suppose that $n\geq 5$ and that $A_{ij}$, for any $i<j$, outputs $v_{ij, r_{ij}}$ for some $r_{ij}$. Then, the list of colors distributed by the sources is either $(c_1,\dots, c_n)=(1,\dots, n)$ or $(c_1,\dots, c_n)=(n, n-1,\dots, 1)$.
\end{claim}

\begin{proof}
By assumption no party outputs a color match. Then, all the colors appear in $\{c_1, \dots, c_n\}$. In particular, there are $i, j$ with $c_i=1$ and $c_j=2$. As the output of $A_{ij}$ is ambiguous, we must have $c_i+c_j=3\in \{i+j, 2(n+1)-(i+j)\}$ which means that either $i+j=3$ or $i+j=2n-1$. Then, there are only two choices for $i, j$: either $\{i, j\} = \{1, 2\}$ or $\{i, j\} = \{n, n-1\}$.  By symmetry (change $i\mapsto n+1-i$) we only analyze the first case, where $c_1 + c_2=3$. In this case either $(c_1, c_2)=(1, 2)$ or $(c_1, c_2)=(2, 1)$. We first show that the latter is impossible. 

Suppose that $(c_1, c_2)=(2, 1)$. Since the output of $A_{1,k}$, for $k>2$ is ambiguous,  $c_k$ should be such that $2+c_k\in \{1+k, 2n+1-k\}$.  This means that $c_k\in \{k-1, 2n-1-k\}$. Then, for $k=3$ we have $c_3 \in \{2, 2n-4\}$. $c_3$ cannot be $2$ since there is already a source with color $2$, and we assumed that no party outputs a color match. Moreover, $c_3$ cannot be $2n-4$ since as $n>4$, we have $2n-4>n$ and the number of colors is $n$. Thus, $(c_1, c_2) =(2, 1)$ is not a valid choice.

We now suppose that $(c_1, c_2)=(1, 2)$. Considering the output of $A_{1, k}$, for $k>2$, we find that $1+c_k\in\{1+k, 2n+1-k\}$ or equivalently $c_k\in \{k, 2n-k\}$. Then, using $c_k\leq n$ we obtain $c_k=k$ as desired. Thus, we obtain  $(c_1,\dots, c_n)=(1,\dots, n)$.
\end{proof}

Now define 
\begin{align}\label{eq:def-q-Kn-CM}
q\big(r_{ij}: 1\leq i<j\leq n, t\big),
\end{align}
be the probability that $A_{ij}$, for any $1\leq i<j\leq n$ outputs $\ket{v_{ij, r_{ij}}}$ and the list of colors is 
\begin{align*}
&(c_1,\dots, c_n)=(1,\dots, n) \text{ if } t=0\\
&(c_1,\dots, c_n)=(n, n-1,\dots, 1) \text{ if } t=1,
\end{align*}
\emph{conditioned} on all outputs being ambiguous. By the above claim, $q\big(r_{ij}: 1\leq i<j\leq n, t\big)$ is a valid probability distribution. We claim that this distribution satisfies
\begin{align*}q\big(r_{ij}: 1\leq i<j\leq n\big) = \frac{1}{2} \Big|\prod_{i<j} \omega_{ij, r_{ij}}^{(i, j)} +\prod_{i<j} \omega_{ij, r_{ij}}^{(n+1-i, n+1-j)}  \Big|^2
\end{align*}
and for any $i<j$
\begin{align*}q(r_{ij}, t) = 
\begin{cases}\frac{1}{2} \Big|\omega_{ij, r_{ij}}^{(i, j)}\Big|^2  &t=1, \\
\frac{1}{2}\Big|\omega_{ij, r_{ij}}^{(n+1-i, n+1-j)}  \Big|^2 ~ &t=2.
\end{cases}
\end{align*}
These are in parallel with Claim~\ref{claim:RingScenario} of the main text used before with similar proof ideas.  

The first equation is essentially a consequence of Claim~\ref{claim:CMKn}. For simplicity of presentation we prove the second equation for $i=1, j=2$, and $t=1$, the general case being similar. We compute:
\begin{align*}
q(r_{12}&,  t=1) \\
&=   \frac{n^n}{2} \Pr(a_{12} = \ket{v_{12, r_{12}}}, t=1 )\\
&=   \frac{n^n}{2n^{n-2}} \Pr(a_{12} = \ket{v_{12, r_{12}}}, c_1=1, c_2=2 )\\
&=   \frac{n^4}{2} \Pr(a_{12} = \ket{v_{12, r_{12}}}, c_1=c_3=1, c_2=c_4=2 )\\
&=   \frac{n^4}{2} \Pr(a_{12} = \ket{v_{12, r_{12}}}, a_{13}=1, a_{24}=2 )\\
&=\frac{1}{2} \Big|\omega_{12, r_{12}}^{(1, 2)}\Big|^2.
\end{align*}
In general, as the output of $A_{ij}$ depends only on the messages of $S_i, S_j$ and the colors are chosen independently and uniformly, we may change the color of all sources except those of $S_i$ and $S_j$. We do so in such a way that $A_{ii'}$ and $A_{jj'}$, for some $i', j'$ different from $i, j$, output a color match. In this case, the colors of $S_i, S_j$ would be fixed without referring to the value of $t$ and the second equation follows. 

Finally, observe that by Proposition~\ref{propo:RingScenario} of the main text there are choices of $\omega_{ij,r}^{(c_i, c_j)}$ for which a distribution $q\big(r_{ij}: 1\leq i<j\leq n, t\big)$ with the given marginals does not exit. 
Thus, the quantum CM distribution is nonlocal. 

\end{appendix}

\end{document}